\documentclass[journal]{IEEEtran}
\usepackage{amsthm}
\usepackage{amssymb}
\usepackage{graphicx}
\usepackage{amsmath}
\usepackage{bm}
{}
{}
\newtheorem{lemma}{Lemma}{}
{}
{}

\usepackage{algorithm}
\usepackage{algpseudocode}
\usepackage{graphics}
\usepackage{epsfig}
\usepackage{multirow}
\usepackage{array}
\usepackage{cite}
\usepackage{stfloats}
\usepackage{subfigure}
\usepackage{booktabs}
\usepackage{color}
\usepackage{tabu} 
\usepackage{stfloats}
\usepackage{mathtools}
\usepackage{stmaryrd} 

\DeclareMathOperator*{\minimize}{minimize}

\begin{document}
\title{Target Sensing in Full-Dimensional MIMO-OFDM ISAC Systems Via Coupled Tensor Decomposition}

\author{\IEEEauthorblockN{
	Ruoyu Zhang, Tingjie Yang, Tianyu Zhao, Lei Cheng, Shuai Wang, Wen Wu, \\ 
    Hüseyin Arslan, \IEEEmembership{Fellow,~IEEE}}
	\thanks{
		R. Zhang, T. Yang, and W. Wu are with the Key Laboratory of Near-Range RF Sensing ICs \& Microsystems (NJUST), Ministry of Education, School of Electronic and Optical Engineering, Nanjing University of Science and Technology, Nanjing 210094, China (e-mail: ryzhang19@njust.edu.cn; 125104022810@njust.edu.cn; wuwen@njust.edu.cn).
        T. Zhao is with School of Communication and Information Engineering, School of Integrated Circuits, Nanjing Institute of Technology, Nanjing 211167, China (e-mail: zhaotianyu@njit.edu.cn).
        L. Cheng is with the College of Information Science and Electronic Engineering, Zhejiang University, Hangzhou 310027, China (e-mail: lei\_cheng@zju.edu.cn).  
        S. Wang is with the Shenzhen Institutes of Advanced Technology, Chinese Academy of Sciences, Shenzhen 518055, China (e-mail: s.wang@siat.ac.cn).
        Hüseyin Arslan is with the Department of Electrical and Electronics Engineering, Istanbul Medipol University, 34810 Istanbul, Türkiye (e-mail: huseyinarslan@medipol.edu.tr).
}
}


\maketitle



%
\IEEEpeerreviewmaketitle


\begin{abstract}
Integrated sensing and communication (ISAC) has been recognized as a key enabler for future wireless communication systems. 
However, existing ISAC designs based on MIMO-OFDM generally fail to fully exploit the spatial degrees of freedom and the communication data payload, resulting in restricted target sensing performance. 
To cope with this issue, in this paper we propose a novel target sensing method for full-dimensional (FD) MIMO-OFDM ISAC systems based on coupled tensor decomposition. 
We first characterize the system model of FD MIMO-OFDM ISAC, where the base station employs a large number of antennas in two dimensions and facilitates both sensing and communication via the data payload of OFDM signals.
Then, we formulate the problem of estimating target parameters as a coupled canonical polyadic decomposition (CPD) problem, which jointly leverages the echo signals of communication data payload over multiple subcarriers in the frequency domain.
The theoretical analysis indicates that the formulated coupled CPD problem possesses a more relaxed uniqueness condition compared with the ordinary CPD problem, thereby enabling more stable target sensing performance through the pure communication data payload in FD MIMO-OFDM ISAC.
In addition, we develop a coupled CPD-based target sensing algorithm, where the target parameters, including azimuth and elevation angles, ranges, and velocities, are extracted from the decomposed factor matrices.
Numerical results show that the proposed target sensing method achieves a higher success rate and more accurate parameter estimation than existing tensor decomposition methods.
\end{abstract}


\section{Introduction}
Integrated sensing and communication (ISAC) has been recognized as a pivotal technology for sixth-generation (6G) wireless networks, driven by the escalating demands of emerging intelligent applications such as autonomous driving, low-altitude economy, and smart cities 
\cite{Liu2022Integrated,Zhang2025Terahertz,Luo2025ISAC}. 
Unlike conventional designs that treat sensing and communication as isolated functions, ISAC aims to unify these two fundamental capabilities into a single system, thereby substantially enhancing spectral utilization as well as hardware and energy efficiency \cite{Zhou2022Integrated}. 
Attracted by these benefits, researchers from both academia and industry have recently been exploring potential ISAC implementations for forthcoming 6G wireless networks\cite{Ali2020Leveraging,Zhang2021Perceptive}.

Orthogonal frequency division multiplexing (OFDM) 
has been recognized as a highly attractive ISAC waveform, due to its potential to provide excellent sensing and communication performance as well as backward compatibility. From a communication perspective, OFDM is widely employed in modern wireless standards, such as 5G NR and IEEE 802.11 
\cite{Yli2021Frequency,Rahbari2017Exploiting}, 
due to its high spectral efficiency and implementation flexibility. 
From a sensing perspective, the multi-carrier structure of OFDM ensures the independence and orthogonality of spectral components, enabling decoupled estimation of target range and velocity in separate dimensions
\cite{Sturm2011Waveform,Liyanaarachchi2021Optimized,Zhang2022Integrated,Zhang2025Target}.
Meanwhile, multiple-input multiple-output (MIMO), along with its advanced evolutions such as massive MIMO and full-dimensional (FD) MIMO \cite{larsson2014massive,Stoica2009MIMO,Yu2020Performance,Kim2014Full}, 
can provide additional spatial degrees of freedom (DoFs) to facilitate beamforming gains for enhanced coverage, spatial multiplexing for high-capacity communications, and waveform diversity for improved sensing resolution and parameter estimation performance.
Consequently, benefiting from the spatial DoFs and the endogenous communication and sensing capabilities, 
the synergy of MIMO and OFDM (MIMO-OFDM) is regarded as one of the most promising approaches for ISAC deployment in practical wireless networks.

The crucial challenge for MIMO-OFDM ISAC lies in how to effectively achieve target sensing while satisfying the performance requirements for wireless communications. 
To achieve this goal, various ISAC waveform and beamforming designs have been proposed based on different communication and sensing performance metrics \cite{Liu2018Toward,liuxiang2020joint,Liu2022Cramér,Hua2023Optimal,Li2025MIMO}. 
The authors in \cite{Liu2018Toward} developed a globally optimal waveform solution for the omnidirectional and directional beampattern design problems in MIMO ISAC systems. 
Subsequently, the typical sensing performance metrics, such as the mean squared error (MSE) between the actual and desired beampattern \cite{liuxiang2020joint}, the Cramér-Rao bound (CRB) \cite{Liu2022Cramér}, the beampattern gain \cite{Hua2023Optimal}, and the integrated sidelobe level \cite{Li2025MIMO}, are taken into account for joint beamforming design in MIMO-OFDM ISAC systems. 
Meanwhile, common communication metrics including the achievable rate \cite{chen2022generalized}, 
the mutual information \cite{Wei2024Waveform,Chen2026Movable},
and energy efficiency \cite{Zou2024Energy} are also utilized for ISAC design in MIMO systems. 
These ISAC works have integrated target sensing and communication functionalities into one platform, but generally assume partially known target-parameter information, such as angles, ranges, and velocities.

To actually estimate the unknown target parameters, a series of super-resolution and high-accuracy target sensing methods have also been proposed over the past years, such as multiple signal classification (MUSIC) \cite{Stoica1989MUSIC}, and compressive sensing (CS)-based methods \cite{yang2013off,Zhang2022DOA}. 
These existing techniques can be readily applied and extended to MIMO-OFDM ISAC systems \cite{keskin2021mimo,ercan2025ris,Ni2021Uplink,Xiao2024A,Liu2020Joint,Gao2023ISACmassiveMIMO,rahman2019framework,Shafin2016DoA}. 
For example, 
the authors in \cite{keskin2021mimo} proposed to estimate the delay, Doppler, and angle parameters of multiple targets via the orthogonal matching pursuit (OMP) procedure. 
Also, the authors in \cite{ercan2025ris} investigated the reconfigurable intelligent surface (RIS)-aided monostatic multi-target sensing problem with angle-Doppler coupling and developed an OMP-based method with successive target detection and interference cancellation.
The authors in \cite{Ni2021Uplink} proposed a high-resolution direction-of-arrival (DOA) estimation algorithm for uplink perceptive mobile networks.
As a further step, the authors in \cite{Xiao2024A} proposed a joint angle-range-velocity estimation algorithm for MIMO-OFDM ISAC systems. 
For downlink sensing, the authors in \cite{Liu2020Joint} proposed to leverage omnidirectional pilot signals to actively sense the targets, where typical MUSIC and matched-filtering algorithms are employed. 
Gao et al. \cite{Gao2023ISACmassiveMIMO} proposed a CS-based target parameter estimation algorithm for millimeter-wave (mmWave) MIMO-OFDM ISAC systems.
To further exploit the spatial DoFs, the authors in \cite{Shafin2016DoA} proposed a DOA estimation method for FD MIMO-OFDM systems, where ergodic capacity analysis accounting for DOA estimation errors is also conducted.
However, both MUSIC and CS algorithms neglect the inherent high-dimensional characteristics in the received signals, which results in unsatisfactory estimation performance and high computational complexity.

Recently, tensor-based signal processing has attracted considerable attention due to its ability to exploit multidimensional data structures, and has been widely applied in target sensing and communication systems, such as DOA estimation, target parameter estimation, and communication channel estimation \cite{Cheng2015Subspace,Shi2021Nested,Yao2026Tensor,Wu2026Enhanced,zhou2017low,Zhang2022Tensor,Zhang2024Channel}.
For instance, 
the authors in \cite{Shi2021Nested} addressed the joint direction of departure and DOA estimation problem with bistatic MIMO radar using tensor decomposition. 
In the context of communication systems, the authors in \cite{qian2019algebraic} developed a canonical polyadic decomposition (CPD)-based channel estimation scheme for frequency division duplex massive MIMO systems, 
while the authors in \cite{Zhou2019FD} developed an improved channel estimation scheme by leveraging the intrinsic tensor feature in FD MIMO channels. 
In addition, Asim et al. \cite{Asim2025Structured} developed a structured tensor-based channel estimation method for RIS-assisted THz communications by exploiting the inherent geometric structure of the cascaded channel.
In \cite{Asim2025Two}, the authors further proposed a two-dimensional channel parameter estimation scheme for intelligent reflection surface-assisted networks by combining pilot decoupling with tensor-structured rank-one approximations.
Building upon these advances, we further established a unified tensor framework for channel estimation and target parameter estimation problems in massive MIMO ISAC systems \cite{Zhang2024TWC} \cite{Zhang2024IOTJChannel}, where both problems are formulated within a CPD model, enabling efficient estimation of angular, delay, and Doppler parameters.
The aforementioned studies predominantly adopt a tensor modeling approach to perform target sensing based on either specialized radar waveforms or communication pilot signals, overlooking the potential use of data payload signals for sensing purposes. 
In practice, data payload signals constitute the majority of transmitted signals in MIMO-OFDM ISAC systems \cite{Lu2025Sensing}, yet their exploitation for target sensing remains underexplored. 
Recent studies have begun to examine the use of communication data payload signals for sensing tasks. 
Particularly, the authors in \cite{Liu2025CPOFDM} demonstrated that OFDM data payload signals can achieve the lowest ranging sidelobe levels under quadrature amplitude modulation (QAM) and phase-shift keying. Meanwhile, the authors in \cite{Lu2024Random} analyzed sensing performance under random signaling and proposed a precoding design method for multi-antenna ISAC systems.
However, these works primarily focus on single-antenna OFDM ISAC or MIMO ISAC precoding design, without fully exploiting the multidimensional structure of the received echoes.
How to effectively harness the echo signals of data payload for target sensing and accurately extract unknown target parameters within a tensor framework remains a challenging and open problem.

In this paper, we address the target sensing problem for full-dimensional MIMO-OFDM ISAC systems within a coupled tensor decomposition framework, where the target parameters are fully extracted from the echo signals of the communication data payload. The main contributions of the paper are summarized as follows: 

\begin{itemize}
	\item 
	We make the first attempt to characterize the system model of FD MIMO-OFDM ISAC, where the base station employs a large number of antennas in two dimensions and facilitates both sensing and communication via the data payload of OFDM signals. The same time-frequency resources of OFDM signals are shared for data transmission and are simultaneously reused for sensing, such that target sensing can be enabled without compromising communication functionality.

	\item 
	By exploiting the received echo signals of the communication data payload across multiple subcarriers, we propose to formulate the target sensing problem as a coupled CPD (CCPD) problem. On this basis, we propose a two-stage coupled CPD-based target sensing algorithm, such that the target parameters, including azimuth and elevation angles, time delay, Doppler shift, and reflection coefficients, can be separately extracted from the decomposed factor matrices.

	\item 
	We analyze the uniqueness condition of the proposed coupled CPD formulation for the target sensing problem in FD MIMO-OFDM ISAC. Different from the ordinary CPD formulation, where only one subcarrier can be exploited for target sensing, the proposed coupled CPD formulation possesses a more relaxed uniqueness condition, thereby enabling more stable target-parameter acquisition performance through the pure communication data payload.

	\item 
	Via extensive numerical simulations, we show that compared with the ordinary CPD-based and existing MUSIC-based methods, the proposed coupled CPD-based method can achieve a higher target sensing success rate. Moreover, more accurate target parameter estimation demonstrates the superiority of the proposed target sensing method in terms of communication spectral efficiency.

\end{itemize}

We would like to emphasize the key differences between the current work and our prior works \cite{Zhang2024TWC,Zhang2024IOTJChannel}. Specifically, the works in \cite{Zhang2024TWC,Zhang2024IOTJChannel} considered target sensing during the channel estimation phase, where the sensing functionality entirely relied on predefined pilot sequences and the received echo signal could be naturally formulated as a CPD tensor model. 
In contrast, in this paper we attempt to advance ISAC by performing sensing during the data transmission phase within an FD MIMO-OFDM architecture equipped with planar arrays. By simultaneously reusing communication data payloads across multiple OFDM subcarriers, target sensing is seamlessly enabled without compromising communication functionality. This target sensing problem necessitates the development of a novel CCPD framework, which possesses unique structures that have not been analyzed or exploited in \cite{Zhang2024TWC} and \cite{Zhang2024IOTJChannel}.

Notations:  
Vectors and matrices are denoted by lowercase and uppercase boldface letters, respectively. 
$\left( \cdot \right)^{T}$, 
$\left( \cdot \right)^{*}$, 
$\left( \cdot \right)^{H}$, 
and $\left( \cdot \right)^{-1}$
represent the transpose, conjugate, conjugate transpose, and Moore-Penrose pseudo-inverse operations, respectively.
$\mathbf{I}_{M}$ is the $M \times M$ identity matrix and $\mathrm{j} = \sqrt{-1}$ denotes the imaginary unit. 
$\left\| \cdot \right\|_{2}$ and $\left\| \cdot \right\|_{F}$ denote the $l_{2}$ norm and Frobenius norm, respectively. 
$\odot$, $\otimes$, and $\circ$ denote the Khatri-Rao product, Kronecker product, and outer product, respectively. 
$\mathcal{CN}(0, \sigma^2)$ denotes the zero-mean circularly symmetric complex Gaussian distribution with variance $\sigma^2$. 
$\mathbb{E}\{\cdot\}$ and $\mathrm{Tr}\{\cdot\}$ denote the expectation and trace operations, respectively. 
$[\mathbf{a}]_{m}$ and $[\mathbf{A}]_{m,n}$ denote the $m$-th entry and $(m,n)$-th entry of $\mathbf{a}$ and $\mathbf{A}$, respectively. 
$\mathrm{diag}(\mathbf{a})$ denotes a diagonal matrix formed by $\mathbf{a}$. 
The $j$-th compound matrix of $\mathbf{A}\in\mathbb{C}^{I\times R}$ is denoted by $\mathrm{C}_{j}(\mathbf{A})\in\mathbb{C}^{\mathrm{C}_{I}^{j}\times \mathrm{C}_{R}^{j}}$, where $\mathrm{C}_{m}^{j}=\frac{m!}{j!(m-j)!}$. It is the matrix containing the determinants of all $j\times j$ submatrices of $\mathbf{A}$, arranged with the submatrix index sets in lexicographic order. 
The Khatri-Rao product of two matrices $\mathbf{A} \in \mathbb{C}^{I \times K}$ and $\mathbf{B} \in \mathbb{C}^{J \times K}$ is denoted by $\mathbf{A} \odot \mathbf{B} \in \mathbb{C}^{IJ \times K}$, defined as $\mathbf{A} \odot \mathbf{B} = [\mathbf{a}_1 \otimes \mathbf{b}_1, \mathbf{a}_2 \otimes \mathbf{b}_2, \ldots, \mathbf{a}_K \otimes \mathbf{b}_K]$, where $\mathbf{a}_k$ and $\mathbf{b}_k$ are the $k$-th columns of $\mathbf{A}$ and $\mathbf{B}$, respectively.

\section{System Model} 

As shown in Fig. \ref{Fig_SystemModel}, we consider a generic FD MIMO-OFDM system, where the BS transmits the ISAC wireless signals to simultaneously serve $U$ user equipment terminals (UEs) and sense $Q$ surrounding targets.{\footnote{ We use the term “full-dimensional” to specify MIMO systems with large numbers of antennas in the two dimensions of the rectangular arrays \cite{Alkhateeb2017Multi}.}}
The ISAC BS is equipped with a two-dimensional (2D) antenna array of $N_{\mathrm{bs}}$ elements, $N_{\mathrm{bs}} = N_{\mathrm{bs},x} N_{\mathrm{bs},y}$, 
where $N_{\mathrm{bs},x}$ and $N_{\mathrm{bs},y}$ are the numbers of antennas along the vertical and horizontal directions, respectively. 
Unlike conventional linear arrays that restrict beamforming and sensing to a one-dimensional azimuth plane, FD MIMO utilizes 2D rectangular antenna arrays, enabling the system to fully exploit the inherently three-dimensional spatial structure of wireless environments \cite{Alkhateeb2017Multi}.
The echo signals reflected by the targets are collected by a spatially well-separated 2D receive antenna array at the BS{\footnote{Here, ``spatially well-separated'' refers to the electrical separation required for transmit--receive isolation \cite{keskin2021mimo,rahman2019framework}, rather than to an array separation comparable to the target range. Let $d_{\mathrm{tr}}$ denote the separation between the phase centers of the co-oriented transmit and receive arrays, and let $R_q$ denote the distance from their midpoint to the $q$-th target. Under the condition $\frac{d_{\mathrm{tr}}}{R_q}\ll1$, the transmit and receive directions toward the target are identical, and the propagation delay can be expressed as $\frac{2R_q}{c_0}$.}} with $N_{\mathrm{re}} = N_{\mathrm{re},x} N_{\mathrm{re},y}$ elements,
where $N_{\mathrm{re},x}$ and $N_{\mathrm{re},y}$ are the numbers of receive antennas along the vertical and horizontal directions, respectively.
To achieve the desired trade-off between power consumption and hardware cost, we consider the receive antenna array employing the hybrid beamforming (HBF) architecture, where only $M_{\mathrm{re}}$ radio frequency (RF) chains are connected to $N_{\mathrm{re}}$ antennas through a fully-connected phase-shifter network, $M_{\mathrm{re}} \ll N_{\mathrm{re}}$.

\begin{figure}[!tp]
	\centering
	\includegraphics[width=3.4in]
	{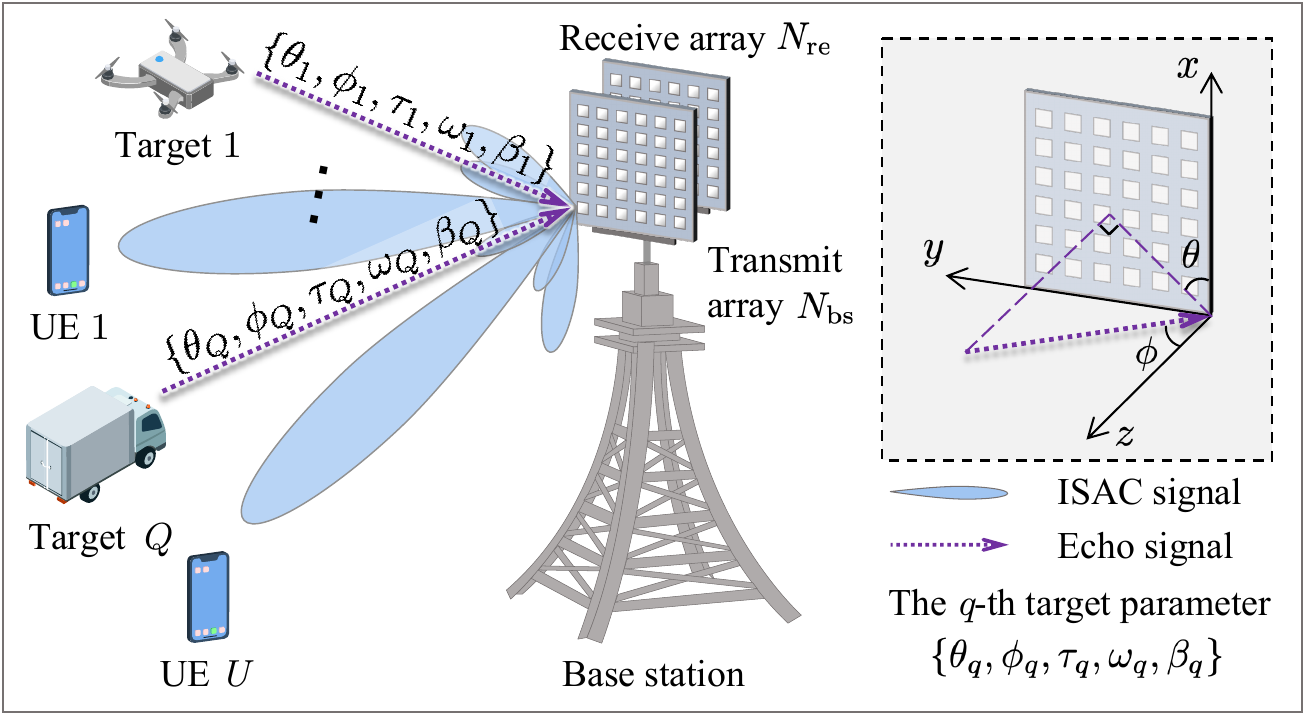}
	\caption{An FD MIMO system model where the base station serves $U$ communication UEs and simultaneously senses $Q$ targets via two-dimensional antenna arrays.}
	\label{Fig_SystemModel}
\end{figure}

An OFDM-based ISAC signal with bandwidth $B$ and $N$ symbols is employed to enable communication information transmission and radar sensing functionalities. 
The complex baseband time-domain transmit signal corresponding to the $n$-th OFDM symbol for the $u$-th user is given by
\begin{align}
s_{n,u}(t) = \frac{1}{\sqrt{K_{0}}} \sum_{k=0}^{K_{0}-1} x_{n,k,u} e^{\mathrm{j} 2 \pi k \Delta f t} \mathrm{rect}\!\left( t - n T_{\mathrm{sym}} \right),
\end{align}
where $x_{n,k,u} \in \mathbb{C}$ denotes the transmitted payload data symbol corresponding to the $u$-th UE at the $k$-th subcarrier of the $n$-th OFDM symbol, $k \in \{0, 1, \ldots, K_{0}-1\}$, $n \in \{0, 1, \ldots, N-1\}$ and 
$K_{0}$ denotes the total number of subcarriers in one OFDM symbol.
$\Delta f = B/K_{0}$ denotes the subcarrier spacing; 
$T_{\mathrm{sym}} = T_{\mathrm{eff}} + T_{\mathrm{cp}}$ is the duration of one OFDM symbol; 
$T_{\mathrm{eff}} = 1/\Delta f$ is the effective OFDM symbol duration; 
$T_{\mathrm{cp}}$ is the duration of the cyclic prefix (CP), which is employed to suppress the intersymbol interference caused by both the maximum time delay of multipath and the time delay corresponding to the maximum detectable range. 
The function $\mathrm{rect}(t)$ is the rectangular window, i.e., $\mathrm{rect}(t) = 1$ for $t \in [-T_{\mathrm{cp}}, T_{\mathrm{eff}}]$, and otherwise $\mathrm{rect}(t) = 0$.
The transmit signal accumulated over the duration of $N$ consecutive OFDM symbols is given by
\begin{align}
\mathbf{s}(t) = \frac{1}{\sqrt{K_{0}}} \sum_{n=0}^{N-1}\sum_{k=0}^{K_{0}-1} \mathbf{P} \mathbf{x}_{n,k} e^{\mathrm{j} 2 \pi k \Delta f t} \mathrm{rect}\!\left( t - n T_{\mathrm{sym}} \right), \label{equ:baseband_st}
\end{align}
where $\mathbf{s}(t) \in \mathbb{C}^{N \times 1}$, 
$\mathbf{x}_{n,k} = [x_{n,k,1}, \ldots, x_{n,k,U}]^{T} \in \mathbb{C}^{U \times 1}$ is the vector of transmitted payload data symbols for $U$ UEs and satisfies $\mathbb{E}\{\mathbf{x}_{n,k} \mathbf{x}_{n,k}^{H}\} = \mathbf{I}_{U}$.
$\mathbf{P} = \mathbf{P}_{\mathrm{RF}}\mathbf{P}_{\mathrm{BB}} \in \mathbb{C}^{N_{\mathrm{bs}} \times U}$ denotes the transmit beamforming matrix, $\mathbf{P}_{\mathrm{RF}}  \in \mathbb{C}^{N_{\mathrm{bs}} \times M_{\mathrm{bs}}}$ is the analog precoder, while $\mathbf{P}_{\mathrm{BB}} \in \mathbb{C}^{M_{\mathrm{bs}} \times U}$ is the digital precoder, and $M_{\mathrm{bs}}$ is the number of transmit RF chains. 
The transmit power constraint is $\mathrm{Tr}\{\mathbf{P} \mathbf{P}^{H}\} = \rho$ and $\rho$ denotes the transmit power budget. 
The baseband OFDM signal in \eqref{equ:baseband_st} is upconverted to the RF band as follows
\begin{align}
\tilde{\mathbf{s}}(t) = \mathbf{s}(t) e^{\mathrm{j} 2 \pi f_{c} t}, \label{equ:baseband2RFband}
\end{align}
where $f_{c}$ is the carrier frequency of the considered ISAC system. 
In this work, we focus on downlink communication and sensing, where the OFDM ISAC signal is shared to support both the sensing and communication functionalities. 
The next subsections develop the target sensing and communication transmission models.

\subsection{Target Sensing Model}

Suppose that the $q$-th target exists in the three-dimensional space at a distance $R_{q}$ with a radial velocity $V_{q}$, $q \in \{1, \ldots, Q\}$. 
Due to the 2D UPA deployed at the BS, 
the direction of the $q$-th target is described by its azimuth angle $\theta_{q}$ and elevation angle $\phi_{q}$, respectively. 
Accordingly, the received RF echo signal over a coherent processing interval (CPI) of duration $N T_{\mathrm{sym}}$ can be expressed as follows \cite{Zhang2024TWC}:
\begin{align} \label{equ:receiveSignalModel_t}
\tilde{\mathbf{y}}(t) \!=\! \sum_{q=1}^{Q} \tilde{\beta}_{q} \mathbf{a}_{\mathrm{re}}(\theta_{q}, \phi_{q})  \mathbf{a}_{\mathrm{bs}}^{T}(\theta_{q}, \phi_{q}) \tilde{\mathbf{s}}(t \!-\! \tau_{q}(t)) + \tilde{\mathbf{v}}^{\mathrm{s}}(t), 
\end{align}
where 
$\tilde{\mathbf{s}}(t)$ is the RF OFDM signal defined in \eqref{equ:baseband2RFband}, 
$\tilde{\beta}_{q}$ denotes the reflection coefficient of the $q$-th target, 
$\tilde{\mathbf{v}}^{\mathrm{s}}(t)$ denotes the propagation noise during the sensing process, 
and the time-dependent delay $\tau_{q}(t)$ is given by
\begin{align}
\tau_{q}(t) = \frac{2(R_{q} - V_{q} t)}{c_{0}} = \tau_{q} - 2 V_{q} t/c_{0}, \label{equ:time_dependent_delay}
\end{align}
where 
$\tau_{q} = \frac{2R_{q}}{c_{0}}$ represents the round-trip time delay and $c_{0}$ denotes the speed of light. 
In \eqref{equ:receiveSignalModel_t},
$\mathbf{a}_{\mathrm{re}}(\theta_{q}, \phi_{q})$ and $\mathbf{a}_{\mathrm{bs}}(\theta_{q}, \phi_{q})$ denote the steering vectors of the receive and transmit arrays, respectively. 
For $\mathbf{a}_{\mathrm{re}}(\theta_{q}, \phi_{q})$, it can be expressed in a compact form as 
\begin{align} 
	{\mathbf{a}}_{\mathrm{re}} \left(\theta_{q}, \phi_{q}\right) \triangleq \mathbf{a}({\mu_{q}; N_{\mathrm{re},x} }) \otimes \mathbf{a}({\nu_{q}; N_{\mathrm{re},y} }) \in \mathbb{C}^{N_{\mathrm{re}} \times 1}, \label{equ:a_re_steering_vec}
\end{align}
where 
$\mathbf{a}({\mu_{q}; N_{\mathrm{re},x} }) \in \mathbb{C}^{N_{\mathrm{re},x} \times 1}$ and 
$\mathbf{a}({\nu_{q}; N_{\mathrm{re},y} }) \in \mathbb{C}^{N_{\mathrm{re},y} \times 1}$ denote the vertical and horizontal steering vectors, respectively, and can be represented as
\begin{align}
\mathbf{a}(\mu_{q}; N_{\mathrm{re},x}) \triangleq \big[ 1, e^{- \mathrm{j} \frac{2 \pi d}{\lambda_c} \mu_{q}} , \ldots, e^{- \mathrm{j} \frac{2 \pi d}{\lambda_c} (N_{\mathrm{re},x}-1) \mu_{q}} \big]^{T},
\end{align}
\begin{align}
\mathbf{a}(\nu_{q}; N_{\mathrm{re},y}) \triangleq \big[ 1, e^{- \mathrm{j} \frac{2 \pi d}{\lambda_c} \nu_{q}} , \ldots, e^{- \mathrm{j} \frac{2 \pi d}{\lambda_c} (N_{\mathrm{re},y}-1) \nu_{q}} \big]^{T},
\end{align}
where $d$ denotes the antenna spacing, $\lambda_{c}$ is the wavelength, and we define the virtual angles 
$\mu_{q} \triangleq \cos\theta_{q} \sin\phi_{q}$, 
and
$\nu_{q} \triangleq \sin\theta_{q} \sin\phi_{q}$.
For $\mathbf{a}_{\mathrm{bs}}(\theta_{q}, \phi_{q})$, its expression is analogous to that of $\mathbf{a}_{\text{re}}(\theta_{q}, \phi_{q})$, i.e., $\mathbf{a}_{\mathrm{bs}}(\theta_{q}, \phi_{q}) \triangleq \mathbf{a}({\mu_{q}; N_{\mathrm{bs},x} }) \otimes \mathbf{a}({\nu_{q}; N_{\mathrm{bs},y} })$, which can be obtained by replacing 
$N_{\text{re}}$, 
$N_{\text{re},x}$, and $N_{\text{re},y}$ 
in \eqref{equ:a_re_steering_vec} with 
$N_{\text{bs}}$,
$N_{\text{bs},x}$, and $N_{\text{bs},y}$, respectively.
The goal of target sensing is to estimate the unknown parameters from the received echo signal. 

At the BS receiver, the echo signal is first down-converted by removing the carrier frequency $e^{\mathrm{j} 2 \pi f_{c} t}$ and then combined via a receive beamforming matrix $\mathbf{Q} \in \mathbb{C}^{N_{\mathrm{re}} \times M_{\mathrm{re}}}$. The resulting baseband echo signal can be expressed as
\begin{align} \label{equ:receiveSignalModel_yt}
\mathbf{y}(t) &\!=\! \mathbf{Q}^{H} \tilde{\mathbf{y}}(t) e^{- \mathrm{j} 2 \pi f_{c} t} \nonumber \\
&\!=\! \sum_{q=1}^{Q}  \tilde{\beta}_{q}  
\mathbf{Q}^{H}
\mathbf{a}_{\mathrm{re}}(\theta_{q}, \phi_{q})  \mathbf{a}_{\mathrm{bs}}^{T}(\theta_{q}, \phi_{q}) \mathbf{s}(t - \tau_{q}(t))  \nonumber \\
& \quad\quad\quad \times e^{- \mathrm{j} 2 \pi f_{c} \tau_{q}(t)}  + \mathbf{v}^{\mathrm{s}}(t) \nonumber\\
&\!\approx\! \sum_{q=1}^{Q} \beta_{q}  
\mathbf{Q}^{H}
\mathbf{a}_{\mathrm{re}}(\theta_{q}, \phi_{q})  \mathbf{a}_{\mathrm{bs}}^{T}(\theta_{q}, \phi_{q}) \mathbf{s}(t - \tau_{q}) \nonumber \\
& \quad\quad\quad \times e^{\mathrm{j} 2 \pi \omega_{q} t} + \mathbf{v}^{\mathrm{s}}(t),	
\end{align}
where $\beta_{q} = \tilde{\beta}_{q}  e^{ - \mathrm{j} 2 \pi f_{c} \tau_{q}}$ is the complex-valued reflection coefficient of the $q$-th target, $\omega_{q} = f_c \frac{2V_{q}}{c_{0}}$ denotes the round-trip Doppler shift, 
$\mathbf{v}^{\mathrm{s}}(t) = \mathbf{Q}^{H} \tilde{\mathbf{v}}^{\mathrm{s}}(t) e^{- \mathrm{j} 2 \pi f_{c} t}$ denotes the combined noise term.
In \eqref{equ:receiveSignalModel_yt}, the second equality holds by substituting \eqref{equ:receiveSignalModel_t}.
The approximation $\mathbf{s}(t - \tau_{q}(t)) \approx \mathbf{s}(t - \tau_{q})$ assumes that the target range migration over the CPI is much smaller than the range resolution, i.e., $|V_{q}| N T_{\mathrm{sym}} \ll \frac{c_{0}}{2B}$ \cite{keskin2021mimo,Zhang2020Joint}.
For the baseband echo signal in \eqref{equ:receiveSignalModel_yt}, 
after CP removal and sampling at 
$t_{s} = (n T_{\mathrm{sym}} + \frac{l}{K_{0} \Delta f})$, $l = 0, 1, \ldots, K_{0}-1$, the discrete-time echo signal corresponding to the $n$-th OFDM symbol is expressed as
\begin{align}
\mathbf{y}_{n}(l) &\!=\! \sum_{q=1}^{Q} \beta_{q}  e^{\mathrm{j} 2 \pi \omega_{q} n T_{\mathrm{sym}}}
\mathbf{Q}^{H}
\mathbf{a}_{\mathrm{re}}(\theta_{q}, \phi_{q})  \mathbf{a}_{\mathrm{bs}}^{T}(\theta_{q}, \phi_{q}) \nonumber \\
&\! \quad \times \!
\frac{1}{\sqrt{K_{0}}} \! \sum_{k=0}^{K_{0}-1} \mathbf{P} \mathbf{x}_{n,k} e^{-\mathrm{j} 2 \pi k \Delta f \tau_{q}} e^{\mathrm{j} 2 \pi k \frac{l}{K_{0}}}  + \mathbf{v}_{n}^{\mathrm{s}}(l), \label{equ:receiveSignalModel_y_n_l}
\end{align}
where $\mathbf{v}_{n}^{\mathrm{s}}(l)$ denotes the noise term of the $l$-th sampling point within the $n$-th OFDM symbol;  
the Doppler effect across different OFDM symbols satisfies $e^{\mathrm{j}2\pi\omega_{q}t} = e^{\mathrm{j}2\pi\omega_{q}nT_{\mathrm{sym}}}$ for $t \in [n T_{\mathrm{sym}}, (n + 1) T_{\mathrm{sym}}]$, since the intrapulse Doppler effect $e^{\mathrm{j}2\pi \frac{l}{K_{0}} \frac{\omega_{q}}{\Delta f}} $ is approximately equal to $1$ and can thus be neglected when $\frac{|\omega_{q}|}{\Delta f} \ll 1$ \cite{Zhang2020Joint}.
From \eqref{equ:receiveSignalModel_y_n_l}, the samples of the received echo signal can be viewed as a radar data cube in the spatial, spectral, and temporal domains, with
$M_{\mathrm{re}}$ RF chains, $K_{0}$ subcarriers, and $N$ OFDM symbols, respectively. 
The goal of target sensing is to acquire the unknown parameters 
from the received echo signal.

\subsection{Communication Transmission Model}

In this subsection, we present the communication transmission model. 
At each UE, the received OFDM signal undergoes CP removal and is subsequently transformed back to the frequency domain via a discrete Fourier transform. 
Accordingly, the received signal at the $u$-th UE can be expressed as\footnote{While we assume single-antenna UEs for simplicity, the communication transmission model can be naturally extended to the case of multi-antenna UEs.}
\begin{align}
	y_{n,k,u} 
	& = 
    \mathbf{h}_{k,u}^{H} \mathbf{P} \mathbf{x}_{n,k} +  z_{n,k,u}^{\mathrm{c}}, \label{equ:communication model_u}
\end{align}
where 
$\mathbf{h}_{k,u} \in \mathbb{C}^{N_{\mathrm{bs}} \times 1}$ is the channel vector at the $k$-th subcarrier associated with the $u$-th UE;
$z_{n,k,u}^{\mathrm{c}}\sim \mathcal{CN}(0, \sigma_{\mathrm{c}}^{2})$ denotes the additive white Gaussian noise with variance $\sigma_{\mathrm{c}}^{2}$.
In \eqref{equ:communication model_u}, the channel vector is explicitly expressed as \cite{Zhang2024IOTJChannel,Gu2024AGeneral}
\begin{align}
	\mathbf{h}_{k,u}
    & =
	\alpha_{u,0} e^{-\mathrm{j} 2 \pi \tau_{u,0} k \Delta f }  \mathbf{a}_{\mathrm{bs}}(\theta_{u,0},\phi_{u,0})\notag \\
    & \:\:\:\: +
    \sqrt{\frac{1}{PK_{f}}}\sum_{p=1}^{P} \alpha_{u,p} e^{-\mathrm{j} 2 \pi \tau_{u,p} k \Delta f }  \mathbf{a}_{\mathrm{bs}}(\theta_{u,p},\phi_{u,p}), 
\end{align}
where $\alpha_{u}$, $\tau_{u}$, $\theta_{u}$, and $\phi_{u}$ denote the coefficient, delay, azimuth angle, and elevation angle of the channel path for the $u$-th UE, respectively; the subscripts $0$ and $p$ denote the line-of-sight (LoS) path and the $p$-th non-line-of-sight (NLoS) path, respectively;
$P$ denotes the number of NLoS paths, $K_{f}$ is the Rician factor representing the energy ratio between the LoS and NLoS components.
$\mathbf{a}_{\mathrm{bs}}(\theta_{u,p}, \phi_{u,p})$ denotes the steering vector of the transmit array. 
$\mathbf{a}_{\mathrm{bs}}(\theta_{u,p}, \phi_{u,p}) \in \mathbb{C}^{N_{\mathrm{bs}} \times 1}$ is the steering vector associated with the $u$-th UE, and it can be expressed in a compact form as 
${\mathbf{a}}_{\mathrm{bs}} \left(\theta_{u,p}, \phi_{u,p}\right) \triangleq \mathbf{a}({\mu_{u,p}; N_{\mathrm{bs},x} }) \otimes \mathbf{a}({\nu_{u,p}; N_{\mathrm{bs},y} })$, where the virtual angles 
$\mu_{u,p} \triangleq \cos\theta_{u,p} \sin\phi_{u,p}$, 
and 
$\nu_{u,p} \triangleq \sin\theta_{u,p} \sin\phi_{u,p}$.
For the $u$-th user, the signal-to-interference-plus-noise ratio (SINR) is given by \cite{liuxiang2020joint,Yu2024Energy}
\begin{align}
	 \mathrm{SINR}_{k,u} = \frac{\left | 
     \mathbf{h}_{k,u}^{H} \mathbf{p}_{u}\right |^{2}}{\sum_{i\neq u}^{U}\left | 
     \mathbf{h}_{k,u}^{H} \mathbf{p}_{i}\right |^{2}+\sigma _{\mathrm{c}}^{2}},
\end{align}
which is a widely used metric for evaluating multiuser communication performance.
Therefore, the spectral efficiency (SE) of the $U$ users can be expressed as
\begin{align}
	\mathrm{SE} &= \frac{1}{K_{0}}\sum_{k=0}^{K_{0}-1}\sum_{u=1}^{U}\log_{2}(1+\mathrm{SINR}_{k,u}).
\end{align}

For the considered FD MIMO-OFDM ISAC system, the same time-frequency resources are shared for performing both communication information transmission and target sensing.  
Since pure communication signals are employed, the communication functionality remains unaffected. 
In the following, we focus on how to utilize the shared OFDM ISAC signals for target parameter estimation within the framework of coupled tensor formulation. 

\section{Coupled Tensor Formulation and Analysis for Target Parameter Estimation}
In this section, we first present the preliminaries on CPD and CCPD models. Then, we propose a coupled tensor formulation for target sensing during communication data payload transmission in FD MIMO-OFDM ISAC systems. 
On this basis, we analyze the uniqueness condition of the proposed coupled tensor formulation for target parameter estimation.

\subsection{Preliminaries on CPD and CCPD Models}
\label{subsec:CPD_CCPD}

Before introducing the proposed coupled tensor framework, we first review tensor representations and the CPD and CCPD models, which serve as the mathematical foundation for the subsequent target sensing formulation.
A tensor is a multidimensional array that naturally represents data exhibiting multilinear structures across multiple domains \cite{kolda2009tensor}. When multidimensional observations arise from the superposition of several latent components, tensor decomposition provides an effective means to uncover the underlying low-dimensional structure. 

Specifically, CPD expresses a tensor as a sum of rank-one tensors, thereby enabling the extraction of latent factors associated with different tensor modes \cite{Zhang2024TWC,Cheng2017Probabilistic}.
Consider a third-order tensor $\bm{\mathcal{X}} \in \mathbb{C}^{I \times J \times K}$ with tensor rank $R$. Its CPD representation is given by
\begin{align}
    \bm{\mathcal{X}}
    &=
    \sum_{r=1}^{R}
    \mathbf{a}_r \circ \mathbf{b}_r \circ \mathbf{c}_r \nonumber\\
    &=
    \llbracket
    \mathbf{A}, \mathbf{B}, \mathbf{C}
    \rrbracket,
\end{align}
where the vectors $\mathbf{a}_r$, $\mathbf{b}_r$, and $\mathbf{c}_r$
are the $r$-th columns of the factor matrices
$\mathbf{A} \in \mathbb{C}^{I \times R}$,
$\mathbf{B} \in \mathbb{C}^{J \times R}$, and
$\mathbf{C} \in \mathbb{C}^{K \times R}$, respectively.
These factor matrices characterize the latent structures embedded in different tensor modes and capture the intrinsic multilinear relationships underlying the observations.
An important property of CPD is that tensor unfoldings can be conveniently expressed using the Khatri–Rao product. For example, the mode-1 unfolding $\mathbf{X}_{(1)} \in \mathbb{C}^{I \times JK}$ is given by $\mathbf{X}_{(1)} = \mathbf{A} \left(\mathbf{C} \odot \mathbf{B}   \right)^T$.

When multiple tensors exhibit partially shared latent structures, they can be modeled within a coupled tensor framework \cite{Sorensen2017Multidimentional, Sorensen2017MultidimensionalII}. 
The corresponding decomposition model is referred to as CCPD, which generalizes CPD to jointly process multiple interconnected or coupled tensors. 
Consider a collection of third-order tensors 
$\{\bm{\mathcal{X}}_k\}_{k=0}^{K-1}$.
A typical CCPD model is expressed as
\begin{align}
    \bm{\mathcal{X}}_k
    &=
    \sum_{r=1}^{R}
    \mathbf{a}_{k,r}
    \circ
    \mathbf{b}_{k,r}
    \circ
    \mathbf{c}_{r}  \nonumber\\
    &=
    \llbracket
    \mathbf{A}_k,
    \mathbf{B}_k,
    \mathbf{C}
    \rrbracket,\quad
    k=0,\ldots,K-1,
\end{align}
where the vectors $\mathbf{a}_{k,r}$, $\mathbf{b}_{k,r}$, and $\mathbf{c}_r$
are the $r$-th columns of the factor matrices
$\mathbf{A}_k \in \mathbb{C}^{I \times R}$,
$\mathbf{B}_k \in \mathbb{C}^{J \times R}$, and
$\mathbf{C} \in \mathbb{C}^{K \times R}$, respectively.
Note that $\mathbf{A}_k$ and $\mathbf{B}_k$ are tensor-specific factor matrices, while $\mathbf{C}$ is a shared factor matrix across all tensors.
Consequently, CCPD not only captures the intrinsic multilinear structure within each tensor but also exploits the coupling relationships among multiple tensors through shared latent factors. 
This coupled tensor decomposition model generally provides enhanced identifiability and improved estimation accuracy compared with independent tensor decompositions.

In the following, we show how to formulate the received echo signals during communication payload transmission as a coupled tensor formulation, facilitating accurate target parameter estimation in FD MIMO-OFDM ISAC systems.

\subsection{Coupled Tensor Modeling for Target Parameter Estimation}

For the echo signal in \eqref{equ:receiveSignalModel_y_n_l}, 
by performing the discrete Fourier transform, we can obtain the frequency-domain echo signal as 
\begin{align} 
	\mathbf{y}_{n,k}
    & = \sum_{q=1}^{Q} \beta_{q} e^{-\mathrm{j} 2 \pi \tau_{q} k \Delta f}  \mathbf{Q}^{H} \mathbf{a}_{\mathrm{re}}(\theta_{q}, \phi_{q})  \mathbf{a}_{\mathrm{bs}}^{T}(\theta_{q}, \phi_{q}) \notag\\
    &\quad \;\times \mathbf{P} \mathbf{x}_{n,k} e^{\mathrm{j} 2 \pi  n \omega_{q} T_{\mathrm{sym}} } + \mathbf{v}_{n,k}^{\mathrm{s}},  \label{equ:ReceivedSignalOriginal} 
\end{align}
where $\mathbf{v}_{n,k}^{\mathrm{s}} \in \mathbb{C}^{M_{\mathrm{re}} \times 1}$ represents the additive noise vector at the $k$-th subcarrier of the $n$-th OFDM symbol. 
By concatenating the echo signals of $N$ OFDM symbols, we have
\begin{align}
	\mathbf{Y}_{k} 
	& =\! \sum_{q=1}^{Q} \gamma_{q,k}
    \mathbf{Q}^{H} \mathbf{a}_{\mathrm{re}}(\theta_{q}, \phi_{q})
	  \mathbf{a}_{\mathrm{bs}}^{T}(\theta_{q}, \phi_{q})  
	\mathbf{S}_{k} \mathbf{\Gamma}(\omega_{q})  + \mathbf{V}_{k} \nonumber\\
	& =\! \sum_{q=1}^{Q} \gamma_{q,k}  \mathbf{b}_{\mathrm{re}}(\theta_{q}, \phi_{q})  \mathbf{b}_{\mathrm{bs},k}^{T}(\mu_{q}, \nu_{q}, \omega_{q}) + \mathbf{V}_{k}, 
	 \label{equ:ReceivedSignalMatrix}
\end{align}
where  
$\mathbf{Y}_{k} = [ \mathbf{y}_{1,k}, \mathbf{y}_{2,k}, \ldots, \mathbf{y}_{N, k}] \in \mathbb{C}^{M_{\mathrm{re}} \times N }$ contains the overall echo signals at the $k$-th subcarrier; 
$\mathbf{\Gamma}(\omega_{q}) \in \mathbb{C}^{N \times N }$ is a diagonal matrix whose $n$-th diagonal element is $e^{\mathrm{j} 2 \pi  n \omega_{q} T_{\mathrm{sym}} }$;  
$\mathbf{S}_{k} =  [ \mathbf{P} \mathbf{x}_{1,k}, \ldots, \mathbf{P} \mathbf{x}_{N,k}] = \mathbf{P} \mathbf{X}_{k} \in \mathbb{C}^{M_{\mathrm{bs}} \times N }$ is the precoded data payload signal, $\mathbf{X}_{k} = [\mathbf{x}_{1,k}, \ldots, \mathbf{x}_{N,k} ]$; 
$\mathbf{V}_{k} = [\mathbf{v}_{1,k}^{\mathrm{s}}, \ldots, \mathbf{v}_{N,k}^{\mathrm{s}}] \in \mathbb{C}^{M_{\mathrm{re}} \times N }$ denotes the noise matrix.
The second equality holds by defining 
$\mathbf{b}_{\mathrm{re}}(\theta_{q}, \phi_{q}) = \mathbf{Q}^{H} \mathbf{a}_{\mathrm{re}}(\theta_{q}, \phi_{q})$, 
$\mathbf{b}_{\mathrm{bs},k}(\mu_{q}, \nu_{q}, \omega_{q}) =  \mathbf{\Gamma}(\omega_{q}) \mathbf{S}_{k}^{T} \mathbf{a}_{\mathrm{bs}}(\theta_{q}, \phi_{q})$, 
and $\gamma_{q,k}$ is defined by
\begin{align}
\gamma_{q,k} = \beta_{q} e^{-\mathrm{j} 2 \pi \tau_{q} k \Delta f}. \label{equ:gamma_q_k}
\end{align}

Note that the echo signals $\mathbf{Y}_{k}$ in \eqref{equ:ReceivedSignalMatrix} involve interactive unknown target parameters, including angles $\{\theta_{q}, \phi_{q}\}$, time delays $\{\tau_{q}\}$, and Doppler shifts $\{\omega_{q}\}$, 
such that $\mathbf{Y}_{k}$ cannot be independently decomposed into different dimensions and cannot naturally conform to the typical tensor model. 
To facilitate the target parameter estimation,
we employ the Kronecker structure-based beamforming matrix \cite{Zhang2025Tensor}, i.e., $\mathbf{Q} = \mathbf{Q}_{x} \otimes \mathbf{Q}_{y} \in \mathbb{C}^{N_{\mathrm{re}} \times M_{\mathrm{re}}}$, to receive and combine the echo signals. The signal $\mathbf{Y}_{k}$ in \eqref{equ:ReceivedSignalMatrix} can then be further expressed as
\begin{align}
	\mathbf{Y}_{k} 
	& \!=\!\! \sum_{q=1}^{Q} \! \gamma_{q,k}  \big( \mathbf{d}_{x}(\mu_{q}) \otimes \mathbf{d}_{y}(\nu_{q}) \!\big)
	\mathbf{b}_{\mathrm{bs},k}^{T}(\mu_{q}, \nu_{q}, \omega_{q}) \!+\! \mathbf{V}_{k}, \label{equ:tensor_formulation_k}
\end{align}
where the equality holds due to $\mathbf{b}_{\mathrm{re}}(\theta_{q}, \phi_{q}) = (\mathbf{Q}_{x} \otimes \mathbf{Q}_{y})^{H} \mathbf{a}_{\mathrm{re}}(\theta_{q}, \phi_{q}) = \mathbf{d}_{x}(\mu_{q}) \otimes \mathbf{d}_{y}(\nu_{q})$ and the Kronecker product property, with
$\mathbf{d}_{x}(\mu_{q}) = \mathbf{Q}_{x}^{H} \mathbf{a}(\mu_{q}; N_{\mathrm{re},x})$ and 
$\mathbf{d}_{y}(\nu_{q}) = \mathbf{Q}_{y}^{H} \mathbf{a}(\nu_{q}; N_{\mathrm{re},y})$, 
in which $\mathbf{Q}_{x} \in \mathbb{C}^{N_{\mathrm{re},x} \times M_{x}}$, 
$\mathbf{Q}_{y} \in \mathbb{C}^{N_{\mathrm{re},y} \times M_{y}}$, and $M_{\mathrm{re}} = M_{x} M_{y}$.

By exploiting the definition of CPD in Section \ref{subsec:CPD_CCPD}, we can observe that the transpose of the received echo signal at the $k$-th subcarrier, i.e., 
$\mathbf{Y}_{k}^{T}$, is a matrix form of a third-order tensor $\bm{\mathcal{Y}}_{k} \in \mathbb{C}^{N \times M_{y} \times M_{x}}$ with rank $Q$. 
The $(n+1, m_{y}, m_{x})$-th entry of $\bm{\mathcal{Y}}_{k}$ corresponds to the $(n+1)$-th row and the $((m_{x}-1)M_{y} +  m_{y})$-th column element of $\mathbf{Y}_{k}^{T}$. 
Thus, for the echo signal at each subcarrier $k$, the third-order tensor $\bm{\mathcal{Y}}_{k}$ is a sum of $Q$ rank-one tensors and admits the following CPD format
\cite{Cheng2017Probabilistic}
\begin{align}
	\bm{\mathcal{Y}}_{k} & = \sum_{q=1}^{Q} \mathbf{b}_{\mathrm{bs},k}(\mu_{q}, \nu_{q}, \omega_{q}) 
	\circ ( \gamma_{q,k} \mathbf{d}_{y}(\nu_{q}) )
	\circ \mathbf{d}_{x}(\mu_{q}) + \bm{\mathcal{V}}_{k} \nonumber\\ 
	& = \llbracket 
	\mathbf{B}_{k}, 	
	\mathbf{C}_{k}, 
	\mathbf{D} \rrbracket  + \bm{\mathcal{V}}_{k}, \label{equ:Coupled CPD Model}
\end{align}
where $\bm{\mathcal{V}}_{k} \in \mathbb{C}^{N \times M_{y} \times M_{x}}$ is the noise tensor associated with $\mathbf{V}_{k}$. 
Three matrices 
$\mathbf{B}_{k} \in  \mathbb{C}^{N \times Q}$,  
$\mathbf{C}_{k} \in  \mathbb{C}^{M_{y} \times Q}$, and 
$\mathbf{D} \in  \mathbb{C}^{M_{x} \times Q}$, are the factor matrices of $\bm{\mathcal{Y}}_{k}$, and are respectively expressed as
\begin{align}
	\mathbf{B}_{k} & \triangleq \big[ \mathbf{b}_{\mathrm{bs},k}(\mu_{1}, \nu_{1}, \omega_{1}), \ldots, \mathbf{b}_{\mathrm{bs},k}(\mu_{Q}, \nu_{Q}, \omega_{Q}) \big],  \label{equ:factor_matrix_1} \\
	\mathbf{C}_{k} & \triangleq \big[ \gamma_{1,k} \mathbf{d}_{y}(\nu_{1}), \ldots, \gamma_{Q,k} \mathbf{d}_{y}(\nu_{Q}) \big], \label{equ:factor_matrix_2} \\
	\mathbf{D} & \triangleq \big[ \mathbf{d}_{x}(\mu_{1}), \ldots, \mathbf{d}_{x}(\mu_{Q}) \big].  \label{equ:factor_matrix_3}
\end{align}

According to \eqref{equ:Coupled CPD Model}, the received echo signal at each subcarrier $k$ can be formulated as a third-order tensor admitting a CPD model. 
A key observation from this model is that the factor matrices \(\mathbf{B}_{k}\) and \(\mathbf{C}_{k}\) vary with the subcarrier index $k$ due to the subcarrier-dependent vector \(\mathbf{b}_{\mathrm{bs},k}\) and the scaling factor $\gamma_{q,k}$, 
while the factor matrix $\mathbf{D}$ remains common across all subcarriers. 
This is because $\mathbf{D}$ depends solely on the spatial directions $\mu_{q}$ via the steering vectors $\mathbf{d}_{x}(\mu_{q})$, which are invariant to subcarrier variations. 
Thus, the received echo signal tensor $\bm{\mathcal{Y}}_{k}$ across multiple subcarriers naturally exhibits a coupled CPD structure, which can be formulated as follows:
\begin{align}
\{\bm{\mathcal{Y}}_{k} = \llbracket \mathbf{B}_{k}, \mathbf{C}_{k}, \mathbf{D} \rrbracket + \bm{\mathcal{V}}_{k} \}_{k=0}^{K-1}, \label{equ:Coupled_CPD} 
\end{align}
where $K \leq K_{0}$ denotes the total number of subcarriers used for target sensing. 
The unknown target parameters, e.g., angles $\mu_{q}, \nu_{q}$, time delays $\tau_{q}$, Doppler shifts $\omega_{q}$, and reflection coefficients $\beta_{q}$, are embedded in the factor matrices, 
and the target sensing problem can be formulated as 
\begin{align} \label{equ:TargetParameterEstimation} 
	\minimize_{ \{\theta_{q}, \: \phi_{q}, \: \omega_{q}, \: \tau_{q}, \: \beta_{q} \}_{q=1}^{Q} } \sum_{k = 0}^{K-1} \Big \Vert \bm{\mathcal{Y}}_{k} - \llbracket \mathbf{B}_{k}, \mathbf{C}_{k}, \mathbf{D} \rrbracket \Big \Vert_{F}^{2}.
\end{align}

The formulation \eqref{equ:TargetParameterEstimation} represents a coupled CPD problem, where the goal is to estimate the unknown target parameters $\{\mu_{q}, \nu_{q}, \omega_{q}, \tau_{q}, \beta_{q} \}_{q=1}^{Q}$ from the subcarrier-dependent factor matrices $\mathbf{B}_{k}$, $\mathbf{C}_{k}$, and the coupled common  factor matrix $\mathbf{D}$. 
By jointly exploiting the received echo signals $\{ \bm{\mathcal{Y}}_{k} \}_{k=0}^{K-1}$ across multiple subcarriers, we can potentially enhance the parameter estimation performance with improved robustness and accuracy.
In the following subsections, we will analyze the uniqueness condition of the proposed coupled tensor formulation for target parameter estimation.

\subsection{Uniqueness Analysis for Ordinary Tensor Formulation}
\label{subsec:UniquenessCPD}

The uniqueness condition of the tensor decomposition is essential for target parameter estimation, ensuring that the decomposed factor matrices incorporate all information about the unknown parameters \cite{Zhang2024TWC,Sorensen2017Multidimentional}.
Before discussing the proposed formulation, we first present the uniqueness conditions of the CPD problem for $\bm{\mathcal{Y}}_{k}$ at each individual subcarrier $k$ in the following Lemma \ref{lemma:CPD}.

\begin{lemma} \label{lemma:CPD}
	Let $\mathbf{B}_{k}\in \mathbb{C}^{N\times Q}$, $\mathbf{C}_{k}\in \mathbb{C}^{M_{y}\times Q}$, and $\mathbf{D}\in \mathbb{C}^{M_{x}\times Q}$ be the factor matrices of $\bm{\mathcal{Y}}_{k} \in \mathbb{C}^{N\times M_{y} \times M_{x}}$.
	If 
	\begin{align}
	\begin{cases}
		\mathbf{D}\enspace \mathrm{has \enspace full \enspace column \enspace rank},\\
		\mathrm{C}_{2}(\mathbf{B}_{k})\odot \mathrm{C}_{2}(\mathbf{C}_{k}) \enspace \mathrm{has \enspace full \enspace column \enspace rank}, 
	\end{cases} \label{equ:ordinaryCPD}
	\end{align}
	then the CPD of $\bm{\mathcal{Y}}_{k}$ is unique, where $\mathrm{C}_{2}(\mathbf{B}_{k})\odot \mathrm{C}_{2}(\mathbf{C}_{k})\in\mathbb{C}^{\mathrm{C}_{N}^{2}\mathrm{C}_{M_{y}}^{2}\times\mathrm{C}_{Q}^{2}}$.
\end{lemma}

Lemma \ref{lemma:CPD} provides a sufficient condition to determine the uniqueness of the ordinary CPD problem. For an arbitrary subcarrier $k$, we have 
\begin{align}
\mathbf{Y}_{(3), k}^{T} = (\mathbf{B}_{k} \odot \mathbf{C}_{k}) \mathbf{D}^{T} + \mathbf{V}_{(3), k}^{T},
\end{align}
where $\mathbf{V}_{(3), k}$ is the corresponding noise matrix. 
In generic cases, the full column rank condition \eqref{equ:ordinaryCPD} is satisfied when $Q\leqslant M_{x}$ and $2Q(Q-1)\leqslant N(N-1)M_{y}(M_{y}-1)$. 
Accordingly, given the number of OFDM symbols and receive RF chains, the maximum number of targets $Q$ that can be recovered is determined. 
Nevertheless, the formulation in Lemma \ref{lemma:CPD} performs tensor decomposition independently for each subcarrier in the frequency domain and fails to exploit all information in the echo signals.

\subsection{Uniqueness Analysis for Coupled Tensor Formulation}

In this subsection, we analyze the uniqueness condition of the coupled tensor formulation by simultaneously employing the echo signals from multiple subcarriers. 
The proposed coupled CPD formulation in \eqref{equ:Coupled CPD Model} is an extension of the ordinary CPD formulation.
The coupled CPD of $\{\bm{\mathcal{Y}}_{k}\}_{k=0}^{K-1}$ given by \eqref{equ:Coupled CPD Model} has the following matrix representation
\begin{align}
\begin{bmatrix}
\mathbf{Y}_{(3), 0}^{T}\\
\vdots\\
\mathbf{Y}_{(3), K-1}^{T}
\end{bmatrix}
\!\!=\!\!
\begin{bmatrix}
(\mathbf{B}_{0}\odot \mathbf{C}_{0}) \\
 \vdots \\
(\mathbf{B}_{K-1}\odot \mathbf{C}_{K-1}) 
\end{bmatrix} \!
\mathbf{D}^{T} 
\!\!+\!\!
\begin{bmatrix}
\mathbf{V}_{(3), 0}^{T}\\
\vdots\\
\mathbf{V}_{(3), K-1}^{T}
\end{bmatrix}\!,
\end{align}
where $\mathbf{B}_{k}$ and $\mathbf{C}_{k}$ are coupled into a single matrix and a common $\mathbf{D}$ is shared. 
Compared with the tensor formulation in Section \ref{subsec:UniquenessCPD}, 
the coupled tensor formulation simultaneously exploits the echo signals from multiple subcarriers and makes it possible to utilize the inherent structure across multiple subcarriers.
The uniqueness condition of coupled CPD is given in the following Lemma \ref{lemma:coupled CPD} \cite{Sorensen2017Multidimentional}.

\begin{lemma} \label{lemma:coupled CPD}
	Consider the coupled CPD of $\bm{\mathcal{Y}}_{k} \in \mathbb{C}^{N\times M_{y} \times M_{x}}$ in \eqref{equ:Coupled_CPD}. If
	\begin{align}
		\begin{cases}
			\mathbf{D}\enspace \mathrm{has \enspace full \enspace column \enspace rank},\\
			\mathbf{G}\enspace \mathrm{has \enspace full \enspace column \enspace rank},\label{condition coupled CPD}
		\end{cases}
	\end{align}
	then the coupled CPD of $\{\bm{\mathcal{Y}}_{k}\}_{k=0}^{K-1}$ is unique, where
    \begin{align}
\mathbf{G}
=\begin{bmatrix}
\mathrm{C}_{2}(\mathbf{B}_{0})\odot \mathrm{C}_{2}(\mathbf{C}_{0}) \\
 \vdots \\
\mathrm{C}_{2}(\mathbf{B}_{K-1})\odot \mathrm{C}_{2}(\mathbf{C}_{K-1})
\end{bmatrix}
\in\mathbb{C}^{K\mathrm{C}_{N}^{2}\mathrm{C}_{M_{y}}^{2}\times\mathrm{C}_{Q}^{2}} .
\end{align}
\end{lemma}

Comparing Lemma \ref{lemma:coupled CPD} with Lemma \ref{lemma:CPD}, we observe that the coupled CPD formulation in \eqref{equ:Coupled_CPD} relaxes the uniqueness condition required for parameter identifiability. 
To be specific, 
Lemma \ref{lemma:CPD} demands that the Khatri–Rao product $\mathrm{C}_{2}(\mathbf{B}_{k})\odot\mathrm{C}_{2}(\mathbf{C}_{k})$ be full column rank for \emph{each} individual subcarrier $k$, 
while Lemma \ref{lemma:coupled CPD} only requires the \emph{aggregated} matrix $\mathbf{G}$, which is constructed by stacking these products across all $K$ subcarriers, to be full column rank. 
This structural difference is beneficial: the coupled CPD model exploits the fact that the spatial factor $\mathbf{D}$ is common across subcarriers, while allowing the temporal/spectral factors $\mathbf{B}_{k}$ and $\mathbf{C}_{k}$ to vary.
Accordingly, the proposed coupled CPD formulation fundamentally differs from the ordinary CPD formulation by extending the tensor decomposition from one arbitrary subcarrier to multiple tensors across $K$ subcarriers that share one common factor matrix. 
The relaxation afforded by the coupled CPD directly offers two key practical advantages for target sensing in FD MIMO-ISAC systems. 
\begin{itemize}
	\item In practical ISAC systems employing communication data payload signals, the signal matrix on a given subcarrier may become ill‑conditioned due to the randomness of OFDM symbols, possibly violating the ordinary CPD condition in Lemma \ref{lemma:CPD}. By jointly processing the echoes from multiple subcarriers, the coupled CPD can average out such randomness, since it is sufficient that the overall stack $\mathbf{G}$ remains well‑conditioned, making target parameter estimation far more robust.
    
    \item The aggregation operation effectively integrates echoes over the broad OFDM bandwidth. The increase in row dimension of $\mathbf{G}$ with $K$ enhances the algebraic diversity of the factorization, which can improve the estimation accuracy and lower the signal-to-noise ratio (SNR) threshold for reliable target parameter recovery. This also naturally aligns with the communication frame structure, where data payload signals are dispersed over multiple OFDM subcarriers.
\end{itemize}

Therefore, the coupled CPD formulation can capture the inherent frequency diversity and spatial consistency of the FD MIMO-OFDM ISAC system. 
By unifying the received echoes into a single coupled CPD model, the proposed coupled tensor framework exploits the inherent structure of the random communication data payload for target sensing, instead of treating signal randomness as a hindrance.
Accordingly, the proposed coupled CPD formulation does not incur any loss of communication spectral efficiency, since it is constructed directly from the received echoes of the OFDM data payload signals over multiple subcarriers, without requiring extra pilot symbols or additional time-frequency resources.

\section{Target Parameter Estimation Algorithm via Coupled CPD}

In this section, we propose a target parameter estimation algorithm within the framework of coupled tensor decomposition. 
The proposed algorithm consists of two stages. 
In the first stage, we describe factor-matrix estimation based on the coupled CPD.
In the second stage, we illustrate how to extract the unknown target parameters by leveraging the estimated coupled factor matrices.

\subsection{Estimation of Coupled Factor Matrices}
In this subsection, we present the details of how to estimate the factor matrices from the tensor $\bm{\mathcal{Y}}_{k}$, $k \in \{0, 1, \ldots, K-1\}$. 
To address this issue, one of the most well-known algorithms for ordinary CPD is the alternating least squares (ALS) method \cite{kolda2009tensor}. 
To be specific,  
for each subcarrier $k$, the factor-matrix estimation problem is formulated as follows:
\begin{align}
\left\lbrace \hat{\mathbf{B}}_{k}, 	\hat{\mathbf{C}}_{k}, \hat{\mathbf{D}} \right\rbrace 
= \underset{
\mathbf{B}_{k}, 	
\mathbf{C}_{k}, 
\mathbf{D}}
{\mathrm{argmin}} \Big \Vert \bm{\mathcal{Y}}_{k} - \llbracket
\mathbf{B}_{k}, 	
\mathbf{C}_{k}, 
\mathbf{D} 
\rrbracket \Big \Vert_{F}^{2}. \label{equ:ALS_formulation}
\end{align}

Although ordinary CPD with the ALS method is efficient for estimating factor matrices, it does not exploit the inherent frequency diversity and spatial consistency of the factor matrices, thus limiting the performance of target sensing.
To make full use of the information in the received echo signals across multiple subcarriers, 
we propose to formulate a coupled CPD problem that jointly decomposes the subcarrier-specific factor matrices $\mathbf{B}_{k}$ and $\mathbf{C}_{k}$ 
together with the shared coupled factor matrix $\mathbf{D}$, i.e.,
\begin{equation} 
\left\lbrace \hat{\mathbf{B}}_{k}, 	\hat{\mathbf{C}}_{k}, \hat{\mathbf{D}} \right\rbrace 
= \underset{
\mathbf{B}_{k}, 	
\mathbf{C}_{k}, 
\mathbf{D}}
{\mathrm{argmin}} \sum_{k} \Big \Vert \bm{\mathcal{Y}}_{k} - \llbracket
\mathbf{B}_{k}, 	
\mathbf{C}_{k}, 
\mathbf{D} 
\rrbracket \Big \Vert_{F}^{2}. \label{equ:CCPD}
\end{equation}

To address the CCPD problem in \eqref{equ:CCPD}, we first optimize the shared factor matrix $\mathbf{D}$ while keeping all subcarrier-specific matrices $\mathbf{B}_{0}, \mathbf{B}_{1}, \ldots, \mathbf{B}_{K-1}$ and $\mathbf{C}_{0}, \mathbf{C}_{1}, \ldots, \mathbf{C}_{K-1}$ fixed. 
Subsequently, for each subcarrier $k$, we perform alternating updates: $\mathbf{B}_{k}$ is refined with $\mathbf{C}_{k}$ and $\mathbf{D}$ held constant, followed by updating $\mathbf{C}_{k}$ while fixing $\mathbf{B}_{k}$ and $\mathbf{D}$. This procedure iterates across all subcarriers, with update steps given as follows:
\begin{align}
	\mathbf{D}^{(i+1)}=\arg\min_{\mathbf{D}}
	\left\|
	\begin{bmatrix}\mathbf{Y}_{(3),0}\\ \vdots\\ \mathbf{Y}_{(3),K-1}\end{bmatrix} \!\!-\! 
    \mathbf{D}\!\!
	\begin{bmatrix}(\mathbf{C}_{0}^{(i)}\odot \mathbf{B}_{0}^{(i)})^{T}\\\vdots\\(\mathbf{C}_{K-1}^{(i)}\odot \mathbf{B}_{K-1}^{(i)})^{T} \!\end{bmatrix} \!
	\right\|_{F}^{2}\!,\label{equ:coupled_ALS1}
\end{align}
\begin{align}
\mathbf{B}_{k}^{(i+1)}=\arg\min_{\mathbf{B}_{k}}\left\| \mathbf{Y}_{(1),k}-\mathbf{B}_{k}(\mathbf{D}^{(i+1)}\odot\mathbf{C}_{k}^{(i)})^{T}\right\|_{F}^{2},\label{equ:coupled_ALS2}
\end{align}
\begin{align}
\mathbf{C}_{k}^{(i+1)}=\arg\min_{\mathbf{C}_{k}}\left\| \mathbf{Y}_{(2),k}-\mathbf{C}_{k}(\mathbf{D}^{(i+1)}\odot\mathbf{B}_{k}^{(i+1)})^{T}\right\|_{F}^{2},\label{equ:coupled_ALS3}
\end{align}
where $\mathbf{B}_{k}^{(i)}$, $\mathbf{C}_{k}^{(i)}$, and $\mathbf{D}^{(i)}$ are the updated factor matrices in the $i$-th iteration. 
For a least-squares problem of the form 
$ \hat{\mathbf{X}}=\arg\min_{\mathbf{X}}\left\| \mathbf{Y}-\mathbf{X}\mathbf{Z}^{T}\right\|_{F}^{2}$, 
the optimal solution is given by 
$\hat{\mathbf{X}}=\mathbf{Y}\mathbf{Z}^{*} (\mathbf{Z}^{T}\mathbf{Z}^{*})^{-1}$. 
Accordingly, the solutions to \eqref{equ:coupled_ALS1}, \eqref{equ:coupled_ALS2}, and \eqref{equ:coupled_ALS3} are given by
\begin{align}
	\mathbf{D}^{(i+1)}
    = & \left(\sum_{k=0}^{K-1}\mathbf{Y}_{(3),k}(\mathbf{C}_{k}^{(i)}\odot \mathbf{B}_{k}^{(i)})^{*}\right)\notag\\
	& \times \left(\sum_{k=0}^{K-1}(\mathbf{C}_{k}^{(i)}\odot\mathbf{B}_{k}^{(i)})^{T}(\mathbf{C}_{k}^{(i)}\odot\mathbf{B}_{k}^{(i)})^{*}\right)^{-1}, \label{equ:D_est} \\
	\mathbf{B}_{k}^{(i+1)}
	= & \mathbf{Y}_{(1),k}(\mathbf{D}^{(i+1)}\odot\mathbf{C}_{k}^{(i)})^{*}\notag\\
	& \times \left((\mathbf{D}^{(i+1)}\odot\mathbf{C}_{k}^{(i)})^{T}(\mathbf{D}^{(i+1)}\odot\mathbf{C}_{k}^{(i)})^{*}\right)^{-1}, \label{equ:B_est} \\
	\mathbf{C}_{k}^{(i+1)}
	= & \mathbf{Y}_{(2),k}(\mathbf{D}^{(i+1)}\odot\mathbf{B}_{k}^{(i+1)})^{*}\notag\\
	& \times \left((\mathbf{D}^{(i+1)}\odot\mathbf{B}_{k}^{(i+1)})^{T}(\mathbf{D}^{(i+1)}\odot\mathbf{B}_{k}^{(i+1)})^{*}\right)^{-1}. \label{equ:C_est} 
\end{align}

By iterating the three steps in an alternating manner,
we can obtain estimates of the factor matrices, i.e., 
$\hat{\mathbf{B}}_{k}$, 
$\hat{\mathbf{C}}_{k}$, 
and 
$\hat{\mathbf{D}}$.
The estimated factor matrices and the true factor matrices satisfy the relation
\begin{align} 
\hat{\mathbf{B}}_{k} & = \mathbf{B}_{k}   \mathbf{\Pi} \mathbf{\Delta}_{k}^{(1)} + \mathbf{E}_{k}^{(1)}, k=0, \ldots, K-1, \label{equ:FactorMatrixEst1} \\
\hat{\mathbf{C}}_{k} & = \mathbf{C}_{k}  \mathbf{\Pi} \mathbf{\Delta}_{k}^{(2)} + \mathbf{E}_{k}^{(2)}, k=0, \ldots, K-1,  \label{equ:FactorMatrixEst2} \\
\hat{\mathbf{D}} & = \mathbf{D}  \mathbf{\Pi} \mathbf{\Delta}^{(3)} + \mathbf{E}^{(3)}, \label{equ:FactorMatrixEst3} 
\end{align}
where $\mathbf{\Delta}_{k}^{(1)}$, $\mathbf{\Delta}_{k}^{(2)}$, and $\mathbf{\Delta}^{(3)}$ are the unknown diagonal scaling matrices due to the coupled tensor decomposition; 
$\mathbf{\Pi} \in \mathbb{C}^{Q \times Q}$ represents a permutation matrix common for every factor matrix; 
$\mathbf{E}_{k}^{(1)}$ , $\mathbf{E}_{k}^{(2)}$, and $\mathbf{E}^{(3)}$ are the associated error matrices of tensor factorization.

\subsection{Estimation of Unknown Parameters}
\label{subsec:TVFS Channel Parameters}

In this subsection, we turn to the extraction of unknown target parameters $\theta_{q}$, $\phi_{q}$, $\tau_{q}$, $\beta_{q}$, and $\omega_{q}$ from the estimated factor matrices $\lbrace \hat{\mathbf{B}}_{k}, \hat{\mathbf{C}}_{k}, \hat{\mathbf{D}} \rbrace$, $k=0, \ldots, K-1$. 
Note that the estimated factor matrices are subject to a permutation ambiguity, denoted by the matrix $\mathbf{\Pi}$, which may alter the order of the estimated targets. 
Fortunately, because this ambiguity applies consistently across all estimated factor matrices, it does not disrupt the pairing relationship among different factor matrices. As a result, parameters belonging to the same target remain correctly aligned within their corresponding columns across all dimensions, enabling the automatic pairing of parameters for each specific target. 
In the following, we demonstrate how to estimate the angle, Doppler, delay, and reflection coefficient parameters.

\subsubsection{Estimation of Angle Parameters}
For the considered target sensing problem in FD MIMO-OFDM ISAC systems, each target has azimuth and elevation angles $\theta_{q}$ and $\phi_{q}$, which are equivalently expressed as $\mu_{q}$ and $\nu_{q}$.
Based on the tensor formulation in \eqref{equ:tensor_formulation_k}, these angles are incorporated in $\hat{\mathbf{C}}_{k}$ and $\hat{\mathbf{D}}$, as shown in \eqref{equ:FactorMatrixEst2} and \eqref{equ:FactorMatrixEst3}. 
Taking advantage of the tensor decomposition, one can estimate $\mu_{q}$ and $\nu_{q}$ in a parallel manner. 
To be specific, for $\mu_{q}$, it is associated with each column of $\hat{\mathbf{D}}$, i.e.,
\begin{align}
	\hat{\mathbf{d}}_{q} = \delta_{q}^{(3)} \mathbf{d}_{x}(\mu_{q}) + \mathbf{e}_{q}^{(3)}, 
\end{align}
where $\hat{\mathbf{d}}_{q}$ is the $q$-th column of the estimated factor matrix $\hat{\mathbf{D}}$, $\delta_{q}^{(3)}$ is the $q$-th diagonal element of the unknown diagonal matrix $\mathbf{\Delta}^{(3)}$, and $\mathbf{e}_{q}^{(3)}$ is the $q$-th column of the associated error matrix $\mathbf{E}^{(3)}$.
Accordingly, we can separately extract the angle $\mu_{q}$ via the correlation-based method with respect to each column of $\hat{\mathbf{D}}$ \cite{zhou2017low}, i.e.,
\begin{align} 
\hat{\mu}_{q} 
= \arg\max_{\mu} \frac{\Big| \hat{\mathbf{d}}_{q}^{H} \mathbf{d}_{x}(\mu) \Big|^{2}}
{\big\| \hat{\mathbf{d}}_{q} \big\|_{2}^{2} 
 \big\| \mathbf{d}_{x}(\mu) \big\|_{2}^{2}}. \label{equ:mu_est}
\end{align}

For the parameter $\nu_{q}$, the factor-matrix relationship in \eqref{equ:FactorMatrixEst2} shows that it is incorporated in the factor matrix $\hat{\mathbf{C}}_{k}$, whose $q$-th column is denoted as 
\begin{align}
\hat{\mathbf{c}}_{q,k} = \delta_{q,k}^{(2)} \gamma_{q,k} \mathbf{d}_{y}(\nu_{q}) + \mathbf{e}_{q,k}^{(2)}, k \in \{0, 1, \ldots, K-1\}, \label{equ:c_q_k}
\end{align}
where $\hat{\mathbf{c}}_{q,k}$ and $\mathbf{e}_{q,k}^{(2)}$ are the $q$-th columns of $\hat{\mathbf{C}}_{k}$ and $\mathbf{E}_{k}^{(2)}$, respectively, and $\delta_{q,k}^{(2)}$ is the $q$-th diagonal element of $\mathbf{\Delta}_{k}^{(2)}$. 
To obtain $\nu_{q}$ from \eqref{equ:c_q_k}, one simple way is to perform correlation-based estimation for each $k$, as in \eqref{equ:mu_est}. However, this method is not able to fully utilize the information in the echo signals from all subcarriers. 
Furthermore, the presence of the unknown complex scalar $\delta_{q,k}^{(2)}$ introduces scaling ambiguity. To this end, we propose to formulate the estimation problem of $\nu_{q}$ as follows
\begin{align}\label{equ:nu_formulation}
    \hat{\nu}_{q} = \arg\min_{\nu_{q}} \sum_{k=0}^{K-1} \left\|\hat{\mathbf{c}}_{q,k} - \delta_{q,k}^{(2)} \gamma_{q,k} \mathbf{d}_{y}(\nu_{q})\right\|_{2}^{2},
\end{align}
which exploits information from $K$ subcarriers. 
For problem \eqref{equ:nu_formulation}, we have the following lemma.
\begin{lemma}\label{lemma:solution_c_q_k}
    The solution of the problem \eqref{equ:nu_formulation} is given by
    \begin{align}
        \hat{\nu}_{q} = \arg\max_{\nu _{q}}\sum_{k=0}^{K-1}\frac{\left | \mathbf{d}_{y}^{H}(\nu _{q})\hat{\mathbf{c}}_{q,k}\right |^{2}}{\left\| \mathbf{d}_{y}(\nu _{q})\right\|_{2}^{2}\left\| \hat{\mathbf{c}}_{q,k}\right\|_{2}^{2}},\label{equ:nu_est}
    \end{align}
    which is the maximum likelihood solution of $\nu_{q}$.
\end{lemma}
\begin{proof}
    See Appendix \ref{proof_c_q_k}.
\end{proof}

Based on the estimated $\hat{\mu}_{q}$,  $\hat{\nu}_{q}$, and the relationships $\mu_{q} \triangleq \cos\theta_{q} \sin\phi_{q}$ and
$\nu_{q} \triangleq \sin\theta_{q} \sin\phi_{q}$, we can obtain the estimates of the azimuth and elevation angles through the elementary trigonometric operations as 
\begin{align}
\hat{\theta}_{q} &= \arctan(\frac{\hat{\nu}_{q}}{\hat{\mu}_{q}}),  \label{equ:theta_est} \\
\hat{\phi}_{q}  &= \arcsin(\sqrt{\hat{\mu}_{q}^{2}+\hat{\nu}_{q}^{2}}). \label{equ:phi_est}
\end{align}

\subsubsection{Estimation of Doppler Parameters}

As shown in \eqref{equ:factor_matrix_1}, the Doppler parameters $\{ \omega_{q} \}$ are incorporated in the factor matrix $\mathbf{B}_{k}$. 
Based on the estimated factor matrix $\hat{\mathbf{B}}_{k}$ in \eqref{equ:FactorMatrixEst1}, its $q$-th column is given by
\begin{align} \label{equ:Sigma_a_theta}
	\hat{\mathbf{b}}_{q,k}
	& = \delta_{q,k}^{(1)} \mathbf{\Gamma}(\omega_{q}) \mathbf{S}_{k}^{T} \mathbf{a}_{\mathrm{bs}}(\theta_{q}, \phi_{q}) + \mathbf{e}_{q,k}^{(1)}, 
\end{align}
where $\delta_{q,k}^{(1)}$ is the $q$-th diagonal element of $\mathbf{\Delta}_{k}^{(1)}$ and $\mathbf{e}_{q,k}^{(1)}$ is the $q$-th column of $\mathbf{E}_{k}^{(1)}$, $k = 0, 1, \ldots, K-1$.
As can be observed from \eqref{equ:Sigma_a_theta}, the unknown parameters $\delta_{q,k}^{(1)}$, $\omega_{q}$, $\theta_{q}$, and $\phi_{q}$ are coupled together, making joint estimation of these four parameters very challenging.
Besides, in the considered FD MIMO system, the need to estimate both azimuth and elevation angles results in an extremely high computational complexity for traditional alternating optimization methods \cite{Zhang2024TWC}.
To deal with the aforementioned difficulties, 
we leverage the estimated azimuth and elevation angles in \eqref{equ:theta_est} and \eqref{equ:phi_est} as prior information and formulate the Doppler shift estimation problem as
\begin{align} 
\hat{\omega}_{q} 
& = \arg\min_{\omega_{q}} 
\sum_{k=0}^{K-1} \Big\lVert \hat{\mathbf{b}}_{q,k}  -  \delta_{q,k}^{(1)} \mathbf{\Gamma}(\omega_{q}) \mathbf{S}_{k}^{T} \mathbf{a}_{\mathrm{bs}}(\hat{\theta}_{q}, \hat{\phi}_{q}) \Big\rVert_{2}^2 \nonumber\\
& = \arg\min_{\omega_{q}} 
\sum_{k=0}^{K-1} \Big\lVert \hat{\mathbf{b}}_{q,k}  - \delta_{q,k}^{(1)} \mathbf{Q}_{k} \mathbf{a}_{\mathrm{do}}(\omega_{q}) \Big\rVert_{2}^2, \label{equ:Doppler_formulation}
\end{align}
where the second equality follows from defining a diagonal matrix $\mathbf{Q}_{k}$, whose diagonal elements are $\mathbf{S}_{k}^{T} \mathbf{a}_{\mathrm{bs}}(\hat{\theta}_{q}, \hat{\phi}_{q})$, and extracting the diagonal elements of $\mathbf{\Gamma}(\omega_{q})$ as $\mathbf{a}_{\mathrm{do}}(\omega_{q})$, which is expressed as
\begin{align} \label{equ:a_doppler}
	\mathbf{a}_{\mathrm{do}}(\omega_{q}) = \Big[1, e^{\mathrm{j} 2 \pi \omega_{q} T_{\mathrm{sym}} }, \ldots, e^{\mathrm{j} 2 \pi  (N-1) \omega_{q} T_{\mathrm{sym}} } \Big]^{T}.
\end{align}

For problem \eqref{equ:Doppler_formulation}, we have the following lemma. 
\begin{lemma}\label{lemma:solution_b_q_k}
	The solution of the problem \eqref{equ:Doppler_formulation} is given by
	\begin{align}
		\hat{\omega}_{q} = \arg\max_{\omega_{q}}\sum_{k=0}^{K-1}\frac{\left | (\mathbf{Q}_{k}^{T}\mathbf{a}_{\mathrm{do}}(\omega_{q}))^{H}\hat{\mathbf{b}}_{q,k}\right |^{2}}{\left\| \mathbf{Q}_{k}^{T}\mathbf{a}_{\mathrm{do}}(\omega_{q})\right\|_{2}^{2}\left\| \hat{\mathbf{b}}_{q,k}\right\|_{2}^{2}},\label{equ:omega_est}
	\end{align}
	 which is the maximum likelihood solution of $\omega_{q}$.
\end{lemma}

\begin{proof}
See Appendix \ref{proof_b_q_k}.
\end{proof}

\subsubsection{Estimation of Delay and Reflection Coefficient Parameters}
We next address the estimation of delay and reflection coefficient parameters. 
Note that these two parameters are incorporated in the factor matrix $\mathbf{B}_{k}$ and coupled in $\gamma_{q,k} = \beta_{q} e^{-\mathrm{j} 2 \pi \tau_{q} k \Delta f}$.
To deal with this issue, 
based on the echo signal tensor model in \eqref{equ:tensor_formulation_k}, by taking the mode-2 unfolding of $\bm{\mathcal{Y}}$ and applying the transpose and vectorization operations, we can obtain the echo signal representation as 
\begin{align}
	\mathrm{vec} \big(\mathbf{Y}_{(2),k}^{T}\big) 
	& = \mathbf{C}_{y} \odot  ( \mathbf{D} \odot \mathbf{B}_{k} ) \bm{\gamma}_{k} + \mathrm{vec} \big(\mathbf{N}_{(2),k}^{T}\big),
\end{align}
where $\bm{\gamma}_{k} = [\gamma_{1,k}, \ldots, \gamma_{Q,k}]^{T}  \in \mathbb{C}^{Q \times 1}$; 
and $\mathbf{C}_{y} = [\mathbf{d}_{y}(\nu_{1}), \ldots, \mathbf{d}_{y}(\nu_{Q})] \in \mathbb{C}^{M_{y} \times Q}$; 
$\mathbf{N}_{(2),k}$ denotes the associated noise term at the $k$-th subcarrier. 
Following the least-squares criterion, the estimate of $\bm{\gamma}_{k}$ can be obtained by 
\begin{align} 
	\hat{\bm{\gamma}}_{k} = \big(\bar{\mathbf{C}}_{y} \odot  ( \bar{\mathbf{D}} \odot \bar{\mathbf{B}}_{k} ) \big)^{-1} \mathrm{vec} \big(\mathbf{Y}_{(2),k}^{T}\big),
\end{align}
where $\bar{\mathbf{C}}_{y}$, $\bar{\mathbf{D}}$, and $\bar{\mathbf{B}}_{k}$ are constructed based on the estimated parameters $\lbrace \hat{\nu}_{q} \rbrace_{q=1}^{Q}$,  $\lbrace \hat{\mu}_{q} \rbrace_{q=1}^{Q}$, and $\lbrace \hat{\omega}_{q} \rbrace_{q=1}^{Q}$.

To this end, explicit estimates of the reflection coefficients and time delays have not been obtained. 
To address this issue, it is necessary to combine the coefficients $\bm{\gamma}_{k}$ of multiple subcarriers. 
From \eqref{equ:gamma_q_k}, we have
\begin{align}\label{equ:gamma_matrix}
	\big[\bm{\gamma}_{0}, \ldots, \bm{\gamma}_{K-1}\big] =  \mathrm{diag}(\beta_{1}, \beta_{2}, \ldots, \beta_{Q}) \mathbf{A}_{\mathrm{td}}^{T} \in \mathbb{C}^{Q \times K},
\end{align}
where the $q$-th column of $\mathbf{A}_{\mathrm{td}} \in \mathbb{C}^{Q \times K}$ is constructed from $\mathbf{a}_{\mathrm{td}}(\tau_{q})$, which is expressed as
\begin{align}
    \mathbf{a}_{\mathrm{td}}(\tau_{q}) = \Big[1, e^{-\mathrm{j} 2 \pi \tau_{q} \Delta f}, \ldots, e^{-\mathrm{j} 2 \pi \tau_{q} (K-1) \Delta f} \Big]^{T}.
\end{align}

From \eqref{equ:gamma_matrix}, we can obtain the following relationship for delay $\tau_{q}$:
\begin{align}
    \bm{\gamma}_{q} = \beta_{q} \mathbf{a}_{\mathrm{td}}(\tau_{q}),
\end{align}
where $ \bm{\gamma}_{q} = [\gamma_{q,0},\ldots,\gamma_{q,K-1}]^{T} $ is the $q$-th column of the matrix $\left[\bm{\gamma}_{0}, \ldots, \bm{\gamma}_{K-1}\right]^{T}$.
Based on the estimated $\left\{ \hat{\bm{\gamma}}_{k} \right\}_{k=0}^{K-1}$ of $K$ subcarriers, we can construct the coefficient matrix $\big[\hat{\bm{\gamma}}_{0}, \ldots, \hat{\bm{\gamma}}_{K-1}\big]^{T}$ and obtain the estimated $\hat{\bm{\gamma}}_{q}$, i.e., $\hat{\bm{\gamma}}_{q} = [\hat{\gamma}_{q,0},\ldots,\hat{\gamma}_{q,K-1}]^{T}$. 
Then, we can estimate the time delay as
\begin{align} \label{equ:tau_est_beamSquint}
	\hat{\tau}_{q} & = \arg\max_{\tau_{q}} \frac{\Big|\hat{\bm{\gamma}}_{q}^{H} \mathbf{a}_{\mathrm{td}}(\tau_{q}) \Big|^{2} }{\left\| \hat{\bm{\gamma}}_{q} \right\|_{2}^{2} \left\|\mathbf{a}_{\mathrm{td}}(\tau_{q})\right\|_{2}^{2}}.
\end{align}  

With the estimated $\left\{\hat{\bm{\gamma}}_{q}\right\}_{q=1}^{Q} $ and $\left\{\hat{\tau}_{q}\right\}_{q=1}^{Q} $, the reflection coefficient is obtained by
\begin{align} \label{equ:beta_est_beamSquint}
	\hat{\beta}_{q} = \mathbf{a}_{\mathrm{td}}^{\dagger}(\hat{\tau}_{q})\hat{\bm{\gamma}}_{q}.
\end{align}

To this end, the unknown five-tuple set of target parameters $\lbrace \hat{\theta}_{q}, \hat{\phi}_{q}, \hat{\omega}_{q}, \hat{\tau}_{q}, \hat{\beta}_{q} \rbrace_{q=1}^{Q}$ can be obtained by the proposed two-stage coupled CPD-based target sensing algorithm, which is summarized in Algorithm \ref{alg:ParaEst}.
Based on the estimated $\left\{\hat{\tau}_{q} \right\}_{q=1}^{Q}$ and $\left\{\hat{\omega}_{q} \right\}_{q=1}^{Q}$, the ranges and velocities of the targets can be obtained by
\begin{align}
    \hat{R}_{q} = \frac{\hat{\tau}_{q} c_{0}}{2}, \\
    \hat{V}_{q} =  \frac{\hat{\omega}_{q} c_{0}}{2 f_c}. 
\end{align}

\begin{algorithm}[!tp]
	\renewcommand{\algorithmicrequire}{\textbf{Input:}}  
	\renewcommand{\algorithmicensure}{\textbf{Output:}}  
	\caption{Proposed Coupled CPD-Based Target Sensing Algorithm}
	\label{alg:ParaEst}
	\begin{algorithmic}[1]
		\Require
		The received echo signals $\mathbf{y}_{n,k}$, $\forall k \in \{0, 1, \ldots, K-1\}$, $\forall n \in \{0, 1, \ldots, N-1\}$, $\mathbf{P}$.

        \For {$i = 1,\ldots, I_{\mathrm{ALS}}$}
		\State Compute the factor matrices $\mathbf{D}^{(i)}$,  $\mathbf{B}_{k}^{(i)}$, $\mathbf{C}_{k}^{(i)}$ via \eqref{equ:D_est}-\eqref{equ:C_est}, respectively;
		\EndFor
        \State Return $\hat{\mathbf{D}} = \mathbf{D}^{(I_{\mathrm{ALS}})} $, $\hat{\mathbf{B}}_{k} = \mathbf{B}_{k}^{(I_{\mathrm{ALS}})} $, and $\hat{\mathbf{C}}_{k} = \mathbf{C}_{k}^{(I_{\mathrm{ALS}})} $;
		\For {each $q \in [1,\ldots,Q]$}
		\State Estimate $\hat{\mu}_{q}$ via \eqref{equ:mu_est};
		\State Estimate $\hat{\nu}_{q}$ via \eqref{equ:nu_est};
		\State Estimate $\hat{\theta}_{q}$ and $\hat{\phi}_{q}$ via \eqref{equ:theta_est} and \eqref{equ:phi_est};
        \State Estimate $\hat{\omega}_{q}$ via \eqref{equ:omega_est};
        \State Estimate $\hat{\tau}_{q}$ via \eqref{equ:tau_est_beamSquint};
        \State Estimate $\hat{\beta}_{q}$ via \eqref{equ:beta_est_beamSquint};
        \EndFor		
		\Ensure
		Return the estimated target parameters $\lbrace \hat{\theta}_{q}, \hat{\phi}_{q}, \hat{\omega}_{q}, \hat{\tau}_{q}, \hat{\beta}_{q} \rbrace_{q=1}^{Q}$.
	\end{algorithmic}
\end{algorithm}

\subsection{Computational Complexity Analysis}

We summarize the computational complexity of the proposed coupled CPD-based target sensing method in Algorithm \ref{alg:ParaEst}. 
Given that the number of ALS iterations is $I_{\mathrm{ALS}}$,
the complexity of estimating factor matrices is 
$\mathcal{O}(I_{\mathrm{ALS}}QKNM_{y}M_{x})$.
Given the numbers of grid points in the one-dimensional searches as $I_{G_{\mu}}$, $I_{G_{\nu}}$, $I_{G_{\omega}}$, and $I_{G_{\tau}}$,
the estimation of $\lbrace\hat{\mu}_{q}\rbrace_{q=1}^{Q}$ and
$\lbrace\hat{\nu}_{q}\rbrace_{q=1}^{Q}$
have complexity orders of
$\mathcal{O}(QI_{G_{\mu}} M_{x})$ and
$\mathcal{O}(Q K I_{G_{\nu}} M_{y})$, respectively.
The complexity of estimating $\lbrace\hat{\theta}_{q}\rbrace_{q=1}^{Q}$ and $\lbrace\hat{\phi}_{q}\rbrace_{q=1}^{Q}$ via \eqref{equ:theta_est} and \eqref{equ:phi_est} is $\mathcal{O}(QI_{G_{\mu}}M_{x})+\mathcal{O}(Q K I_{G_{\nu}}M_{y})$.
The estimation of
$\lbrace\hat{\omega}_{q}\rbrace_{q=1}^{Q}$,
$\lbrace\hat{\tau}_{q}\rbrace_{q=1}^{Q}$, and
$\lbrace\hat{\beta}_{q}\rbrace_{q=1}^{Q}$
have complexity orders of
$\mathcal{O}(Q K I_{G_{\omega}}N)$,
$\mathcal{O}(QI_{G_{\tau}}M_{y})$, and
$\mathcal{O}(2 Q^{2}M_{x}M_{y}N+Q^{3})$, respectively.
Compared with the ordinary CPD method that estimates the factor matrices using only a single subcarrier, the proposed coupled CPD method jointly processes the echo signal tensors from $K$ subcarriers, thereby increasing the computational complexity by a factor of $K$. Thus, the proposed method primarily enhances the parameter estimation performance at the cost of increased receiver-side processing complexity.

\subsection{Cramér-Rao Bound Analysis}

To establish a theoretical performance benchmark for target parameter estimation, we present the Cramér–Rao bound (CRB) for the considered target sensing problem in FD MIMO-OFDM ISAC systems. 
The CRB characterizes the minimum achievable estimation variance of any unbiased estimator and thus provides a fundamental limit for evaluating sensing performance \cite{kay1993fundamentals,Zhang2024IOTJChannel}.
We first define the unknown parameter vector as
\begin{equation}
    \boldsymbol{\xi}
    =
    [
    {\boldsymbol{\theta}^{T},
    \boldsymbol{\phi}^{T},
    \boldsymbol{\omega}^{T},
    \boldsymbol{\tau}^{T},
    \boldsymbol{\beta}_{\rm R}^{T},
    \boldsymbol{\beta}_{\rm I}^{T}}
    ]^{T},
\end{equation}
where
$\boldsymbol{\theta}=[\theta_1,\ldots,\theta_Q]^T$,
$\boldsymbol{\phi}=[\phi_1,\ldots,\phi_Q]^T$,
$\boldsymbol{\omega}=[\omega_1,\ldots,\omega_Q]^T$, 
$\boldsymbol{\tau}=[\tau_1,\ldots,\tau_Q]^T$, 
$\boldsymbol{\beta}_{\rm R}=[{\rm Re}\{\beta_1\},\ldots,{\rm Re}\{\beta_Q\}]^T$ and
$\boldsymbol{\beta}_{\rm I}=[{\rm Im}\{\beta_1\},\ldots,{\rm Im}\{\beta_Q\}]^T$,
denote the azimuth angles, elevation angles, Doppler shifts, time delays, the real and imaginary parts of the reflection coefficients associated with the $Q$ targets, respectively.

For the received echo signal tensor model in \eqref{equ:Coupled_CPD}, since the additive noise components across different subcarriers are mutually independent, 
the log-likelihood function associated with the $K$ subcarriers can be expressed as
\begin{align}
    f(\boldsymbol{\xi })
    & =\text{const}-\frac{1}{{\sigma_{\mathrm{s}}^{2}}}\sum\limits_{k=0}^{K-1}{\left\| \mathbf{Y}_{(1),k}^{T}-(\mathbf{D}\odot {{\mathbf{C}}_{k}})\mathbf{B}_{k}^{T} \right\|_{F}^{2}}, \\ 
    & =\text{const}-\frac{1}{{{\sigma }_{\mathrm{s}}^{2}}}\sum\limits_{k=0}^{K-1}{\left\| \mathbf{Y}_{(2),k}^{T}-(\mathbf{D}\odot {{\mathbf{B}}_{k}})\mathbf{C}_{k}^{T} \right\|_{F}^{2}}, \\ 
    & =\text{const}-\frac{1}{{{\sigma }_{\mathrm{s}}^{2}}}\sum\limits_{k=0}^{K-1}{\left\| \mathbf{Y}_{(3),k}^{T}-({{\mathbf{C}}_{k}}\odot {{\mathbf{B}}_{k}}){{\mathbf{D}}^{T}} \right\|_{F}^{2}}, 
\end{align}
where $\text{const}=-KNM_xM_y\ln(\pi\sigma_{\mathrm{s}}^2)$, 
$\mathbf{Y}_{(1),k}$, $\mathbf{Y}_{(2),k}$, and $\mathbf{Y}_{(3),k}$ denote the mode-1, mode-2, and mode-3 unfoldings of the received tensor corresponding to the $k$-th subcarrier, respectively.

Based on the above likelihood function, the complex Fisher information matrix (FIM) can be calculated as
\begin{align}
    \boldsymbol{\Omega}(\boldsymbol{\xi}) = \mathbb{E}\left\{\left(\frac{\partial f(\boldsymbol{\xi})}{\partial \boldsymbol{\xi}}\right)\left(\frac{\partial f(\boldsymbol{\xi})}{\partial \boldsymbol{\xi}}\right)^H\right\},
\end{align}
which quantifies the amount of information carried by the observations with respect to the unknown parameters. 
According to the different parameter categories, the FIM can be partitioned into the following block matrix form:
\begin{equation}
\boldsymbol{\Omega}(\boldsymbol{\xi})
\!=\!\!
\begin{bmatrix}
{\boldsymbol{\Omega}_{\boldsymbol{\theta\theta}}} &
{\boldsymbol{\Omega}_{\boldsymbol{\theta\phi}}} &
{\boldsymbol{\Omega}_{\boldsymbol{\theta\omega}}} &
{\boldsymbol{\Omega}_{\boldsymbol{\theta\tau}}} &
{\boldsymbol{\Omega}_{\boldsymbol{\theta}\boldsymbol{\beta}_{\rm{R}}}} &
{\boldsymbol{\Omega}_{\boldsymbol{\theta}\boldsymbol{\beta}_{\rm{I}}}} \\
{\boldsymbol{\Omega}_{\boldsymbol{\phi\theta}}} &
{\boldsymbol{\Omega}_{\boldsymbol{\phi\phi}}} &
{\boldsymbol{\Omega}_{\boldsymbol{\phi\omega}}} &
{\boldsymbol{\Omega}_{\boldsymbol{\phi\tau}}} &
{\boldsymbol{\Omega}_{\boldsymbol{\phi}\boldsymbol{\beta}_{\rm{R}}}} &
{\boldsymbol{\Omega}_{\boldsymbol{\phi}\boldsymbol{\beta}_{\rm{I}}}} \\
{\boldsymbol{\Omega}_{\boldsymbol{\omega\theta}}} &
{\boldsymbol{\Omega}_{\boldsymbol{\omega\phi}}} &
{\boldsymbol{\Omega}_{\boldsymbol{\omega\omega}}} &
{\boldsymbol{\Omega}_{\boldsymbol{\omega\tau}}} &
{\boldsymbol{\Omega}_{\boldsymbol{\omega}\boldsymbol{\beta}_{\rm{R}}}} &
{\boldsymbol{\Omega}_{\boldsymbol{\omega}\boldsymbol{\beta}_{\rm{I}}}} \\
{\boldsymbol{\Omega}_{\boldsymbol{\tau\theta}}} &
{\boldsymbol{\Omega}_{\boldsymbol{\tau\phi}}} &
{\boldsymbol{\Omega}_{\boldsymbol{\tau\omega}}} &
{\boldsymbol{\Omega}_{\boldsymbol{\tau\tau}}} &
{\boldsymbol{\Omega}_{\boldsymbol{\tau}\boldsymbol{\beta}_{\rm{R}}}} &
{\boldsymbol{\Omega}_{\boldsymbol{\tau}\boldsymbol{\beta}_{\rm{I}}}} \\
{\boldsymbol{\Omega}_{\boldsymbol{\beta}_{\rm{R}}\boldsymbol{\theta}}} &
{\boldsymbol{\Omega}_{\boldsymbol{\beta}_{\rm{R}}\boldsymbol{\phi}}} &
{\boldsymbol{\Omega}_{\boldsymbol{\beta}_{\rm{R}}\boldsymbol{\omega}}} &
{\boldsymbol{\Omega}_{\boldsymbol{\beta}_{\rm{R}}\boldsymbol{\tau}}} &
{\boldsymbol{\Omega}_{\boldsymbol{\beta}_{\rm{R}}\boldsymbol{\beta}_{\rm{R}}}} &
{\boldsymbol{\Omega}_{\boldsymbol{\beta}_{\rm{R}}\boldsymbol{\beta}_{\rm{I}}}} \\
{\boldsymbol{\Omega}_{\boldsymbol{\beta}_{\rm{I}}\boldsymbol{\theta}}} &
{\boldsymbol{\Omega}_{\boldsymbol{\beta}_{\rm{I}}\boldsymbol{\phi}}} &
{\boldsymbol{\Omega}_{\boldsymbol{\beta}_{\rm{I}}\boldsymbol{\omega}}} &
{\boldsymbol{\Omega}_{\boldsymbol{\beta}_{\rm{I}}\boldsymbol{\tau}}} &
{\boldsymbol{\Omega}_{\boldsymbol{\beta}_{\rm{I}}\boldsymbol{\beta}_{\rm{R}}}} &
{\boldsymbol{\Omega}_{\boldsymbol{\beta}_{\rm{I}}\boldsymbol{\beta}_{\rm{I}}}}
\end{bmatrix}.
\end{equation}

In the calculation of the blocks involving the reflection coefficients, we use
$\partial\mathbf{C}_{k}/\partial\boldsymbol{\beta}_{\rm R}=\mathbf{C}_{\beta,k}$ and
$\partial\mathbf{C}_{k}/\partial\boldsymbol{\beta}_{\rm I}=\mathrm{j} \mathbf{C}_{\beta,k}$, where the $q$-th column of $\mathbf{C}_{\beta,k}$ is
$e^{-\mathrm{j}2\pi\tau_q k\Delta f}\mathbf{d}_{y}(\theta_q,\phi_q)$. 
The CRB matrix is subsequently obtained as the inverse of the FIM, namely,
\begin{equation}
    {\rm CRB}(\boldsymbol{\xi})
    =
    \boldsymbol{\Omega}(\boldsymbol{\xi})^{-1}.
\end{equation}

Since \(\boldsymbol{\xi}\) jointly incorporates all unknown target parameters, the CRBs corresponding to individual parameter categories can be extracted from the associated diagonal blocks of \(\operatorname{CRB}(\boldsymbol{\xi})\). Specifically, following the real-valued parameter order in $\boldsymbol{\xi}$, the root mean square CRBs for the azimuth angles, elevation angles, Doppler shifts, time delays, and reflection coefficients are respectively given by
\begin{align}
 & {\operatorname{CRB}(\boldsymbol{\theta} )=\sqrt{\frac{1}{Q}\operatorname{Tr}\left( {{[\operatorname{CRB}(\boldsymbol{\xi})]}_{1:Q,1:Q}} \right)}}, \\ 
 & {\operatorname{CRB}(\boldsymbol{\phi} )=\sqrt{\frac{1}{Q}\operatorname{Tr}\left( {{[\operatorname{CRB}(\boldsymbol{\xi})]}_{Q+1:2Q,Q+1:2Q}} \right)}}, \\ 
 & {\operatorname{CRB}(\boldsymbol{\omega} )=\sqrt{\frac{1}{Q}\operatorname{Tr}\left( {{[\operatorname{CRB}(\boldsymbol{\xi})]}_{2Q+1:3Q,2Q+1:3Q}} \right)}}, \\ 
 & {\operatorname{CRB}(\boldsymbol{\tau} )=\sqrt{\frac{1}{Q}\operatorname{Tr}\left( {{[\operatorname{CRB}(\boldsymbol{\xi })]}_{3Q+1:4Q,3Q+1:4Q}} \right)}}, \\ 
 & {\operatorname{CRB}(\boldsymbol{\beta}_{\rm R} )=\sqrt{\frac{1}{Q}\operatorname{Tr}\left( {{[\operatorname{CRB}(\boldsymbol{\xi })]}_{4Q+1:5Q,4Q+1:5Q}} \right)}}, \\
 & {\operatorname{CRB}(\boldsymbol{\beta}_{\rm I} )=\sqrt{\frac{1}{Q}\operatorname{Tr}\left( {{[\operatorname{CRB}(\boldsymbol{\xi })]}_{5Q+1:6Q,5Q+1:6Q}} \right)}}.
\end{align}

These CRBs establish the theoretical lower bounds on the estimation errors of the corresponding target parameters and will serve as performance benchmarks for evaluating the effectiveness of the proposed CCPD target sensing method in the subsequent numerical results.

\section{Simulation Results}
\label{Sec:Simulation}

In this section, we present numerical simulations to investigate the performance of the proposed CCPD target sensing method in FD MIMO-OFDM ISAC systems. 
The ISAC BS is equipped with an FD MIMO array, which contains $N_{\mathrm{bs}} = 256$ transmit antennas with $N_{\mathrm{bs},x} = N_{\mathrm{bs},y} = 16$. 
For the receive array, we employ the HBF architecture with $N_{\mathrm{re}} = 256$ receive antennas and $M_{\mathrm{re}} = 16$ RF chains, where $N_{\mathrm{re},x} = N_{\mathrm{re},y} = 16$. 
For the combiners $\mathbf{Q}_{x}$ and $\mathbf{Q}_{y}$, each element is uniformly drawn from the unit circle.
The carrier frequency is $f_{c} = 28$ GHz, the bandwidth of OFDM signals is $B = 100$ MHz with $K_{0} = 128$ subcarriers, and the CP duration is set to $T_{\mathrm{cp}} = \frac{1}{2\Delta f}$.
The power spectral density of the noise is $-174$ dBm/Hz.
The BS serves $U=4$ UEs and simultaneously detects $Q = 3$ targets. 
The conjugate beamforming $\mathbf{P}$ is generated to point toward the UE directions,\footnote{
Under the hybrid beamforming architecture, the transmit beamforming matrix is expressed as $\mathbf{P}=\mathbf{P}_{\rm RF}\mathbf{P}_{\rm BB}$.
Then $\mathbf{P}_{\rm RF}$ and $\mathbf{P}_{\rm BB}$ can be obtained by the following optimization problem \cite{Qi2022Hybrid}, i.e., $\min_{\mathbf{P}_{\rm RF},\mathbf{P}_{\rm BB}}
    \left\|
    \mathbf{P}-\mathbf{P}_{\rm RF}\mathbf{P}_{\rm BB}
    \right\|_{F}^{2},
    {\rm s.t.}
    |[\mathbf{P}_{\rm RF}]_{m,n}|=1$, 
where $\mathbf{P}$ is designed to point toward the LoS directions of UEs.} and the QAM constellation is adopted for data payload transmission \cite{Liu2020Joint}.
The LoS directions $\{\mu_{u,0}, \nu_{u,0}\}_{u=1}^{U}$ of these UEs are set to
$\{0, 0\}$, $\{0.5, 0.5\}$, $\{0.5, -0.5\}$, $\{-0.5, 0.5\}$, respectively.
The communication channel experiences $P=1$ NLoS path and the Rician factor $K_{\mathrm{f}}$ is $30$ dB. 
The channel coefficients are generated by $\alpha_{u,p} = \frac{\lambda_{c}}{4 \pi d_{u} } e^{\mathrm{j} \vartheta_{u,p}}$, where $d_{u} \sim \mathcal{U}[20, 45]$ denotes the distance between the BS and the UE, and $\vartheta_{u,p}$ is the corresponding phase drawn from $\mathcal{U}[0, 2\pi)$. 
For the target parameters, the azimuth and elevation angles $\theta_{q}$ and $\phi_{q}$ are randomly generated from $[0, 2\pi]$ and $[0, \pi/3]$, respectively, while the range and velocity are uniformly distributed in $[R_{\mathrm{min}}, R_{\mathrm{max}}]$ and  $[-V_{\mathrm{max}}, V_{\mathrm{max}}]$, respectively, with $R_{\mathrm{min}} = 20$ m, 
$R_{\mathrm{max}} = 45$ m, 
and the maximum velocity $V_{\mathrm{max}} = 30$ m/s. 
The reflection coefficient of each target is generated by 
$\beta_{q} = \sqrt{\frac{\sigma_{q}  \lambda_{c}^2}{ (4 \pi)^{3} R_{q}^{4} }} e^{\mathrm{j} \vartheta_{q}}$ \cite{Zhang2024IOTJChannel},
where $\sigma_{q}$ is the radar cross section and is set to $1$ m$^2$, 
$\vartheta_{q}$ denotes a random phase drawn from $\mathcal{U}[0, 2\pi)$. The main simulation parameters used throughout this section are summarized in Table~\ref{tab:simulation_parameters}.
The root mean square error (RMSE) is adopted to evaluate the target parameter estimation performance 
\begin{equation}
\mathrm{RMSE}(\chi) = \sqrt{\frac{1}{Q} \sum_{q=1}^{Q} \big(\hat{\chi}_{q} - \chi_{q} \big)^2 }, \label{eq:rmse}
\end{equation}
where $\hat{\chi}_{q}$ and $\chi_{q}$ denote the estimated and corresponding true parameters, respectively, and $\chi_{q} \in \lbrace \theta_{q}, \phi_{q}, \tau_{q}, \omega_{q} \rbrace$.
In terms of benchmarks for performance comparison, we employ the conventional ordinary CPD-based method \cite{Zhang2024TWC}, the typical multiple signal classification (MUSIC) method \cite{Liu2020Joint}, and the OMP-based multi-target sensing method \cite{ercan2025ris}.

\begin{table}[t]
\centering
\caption{Main Simulation Parameters}
\label{tab:simulation_parameters}
\setlength{\tabcolsep}{3pt}
\renewcommand{\arraystretch}{1.15}
\begin{tabular}{|
>{\centering\arraybackslash}p{0.50\columnwidth}|
>{\centering\arraybackslash}p{0.12\columnwidth}|
>{\centering\arraybackslash}p{0.25\columnwidth}|}
\hline
\textbf{Parameter} & \textbf{Symbol} & \textbf{Value} \\
\hline
Number of BS transmit antennas & $N_{\mathrm{bs}}$ & $256$ \\
\hline
Number of BS receive antennas  & $N_{\mathrm{re}}$ & $256$ \\
\hline
Number of transmit RF chains & $M_{\mathrm{bs}}$ & $4$ \\
\hline
Number of receive RF chains & $M_{\mathrm{re}}$ & $16$ \\
\hline
Carrier frequency & $f_c$ & $28\,\mathrm{GHz}$ \\
\hline
Bandwidth & $B$ & $100\,\mathrm{MHz}$ \\
\hline
Number of subcarriers & $K_0$ & $128$ \\
\hline
Subcarrier spacing & $\Delta f$ & $B/K_0$ \\
\hline
CP duration & $T_{\mathrm{cp}}$ & $1/(2\Delta f)$ \\
\hline
Number of UEs & $U$ & $4$ \\
\hline
Number of targets & $Q$ & $3$ \\
\hline
Number of NLoS paths & $P$ & $1$ \\
\hline
Rician factor & $K_{\mathrm{f}}$ & $30\,\mathrm{dB}$ \\
\hline
Distance between the BS and UEs & $d_u$ & $\mathcal{U}[20, 45]\,\mathrm{m}$ \\
\hline
Distance between the BS and targets & $R_q$ & $\mathcal{U}[20,45]\,\mathrm{m}$ \\
\hline
Velocity of targets & $V_q$ & $\mathcal{U}[-30,30]\,\mathrm{m/s}$ \\
\hline
Azimuth angle of targets  & $\theta_q$ & $\mathcal{U}[0,2\pi]$ rad\\
\hline
Elevation angle of targets  & $\phi_q$ & $\mathcal{U}[0,\pi/3]$ rad \\
\hline
Radar cross section of targets & $\sigma_q$ & $1\,\mathrm{m}^2$ \\
\hline
\end{tabular}
\end{table}

\begin{figure}[!tp]
	\centering
	\includegraphics[width=3.2in]{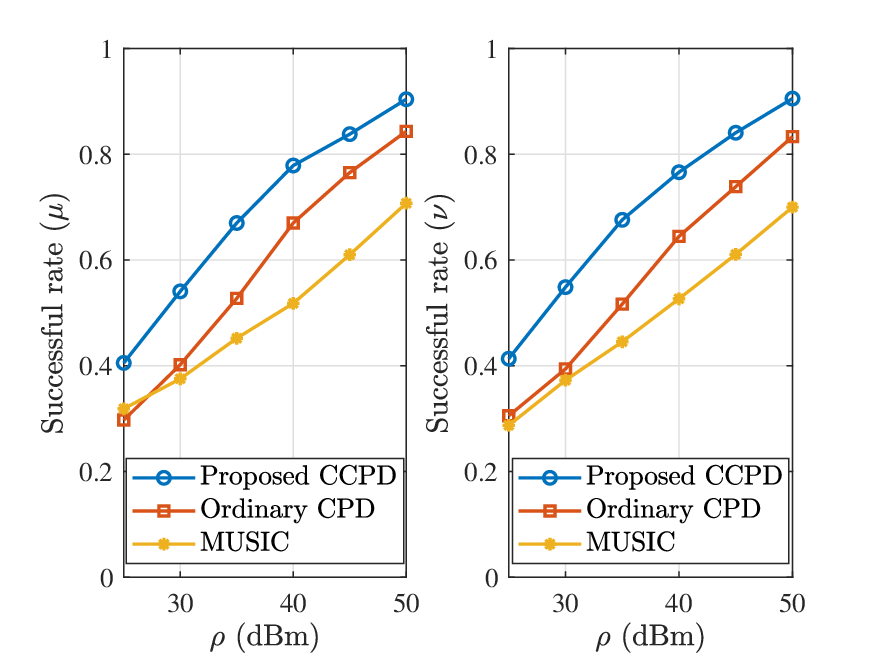}
	\caption{The success rate of correctly estimating $\mu$ and $\nu$, where $K=8$ and $N=16$.}
	\label{Fig_success_rho}
\end{figure}

\begin{figure}[!tp]
	\centering
	\includegraphics[width=3.2in]{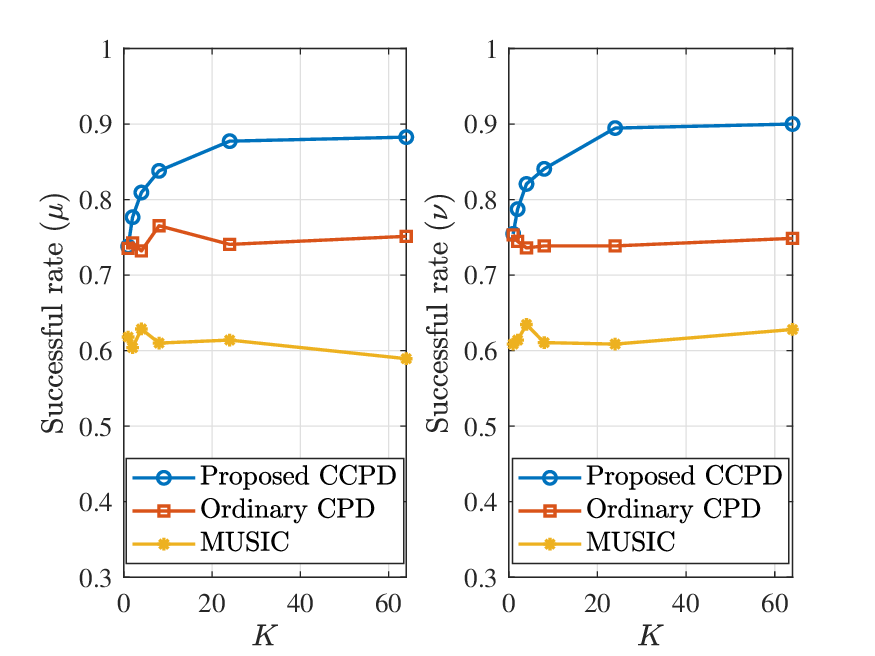}
	\caption{The success rate of correctly estimating $\mu$ and $\nu$ versus the number of subcarriers $K$, where $\rho = 45$ dBm and $N = 16$.}
	\label{Fig_success_K}
\end{figure}

\begin{figure}[!tp]
	\centering
	\includegraphics[width=3.2in]{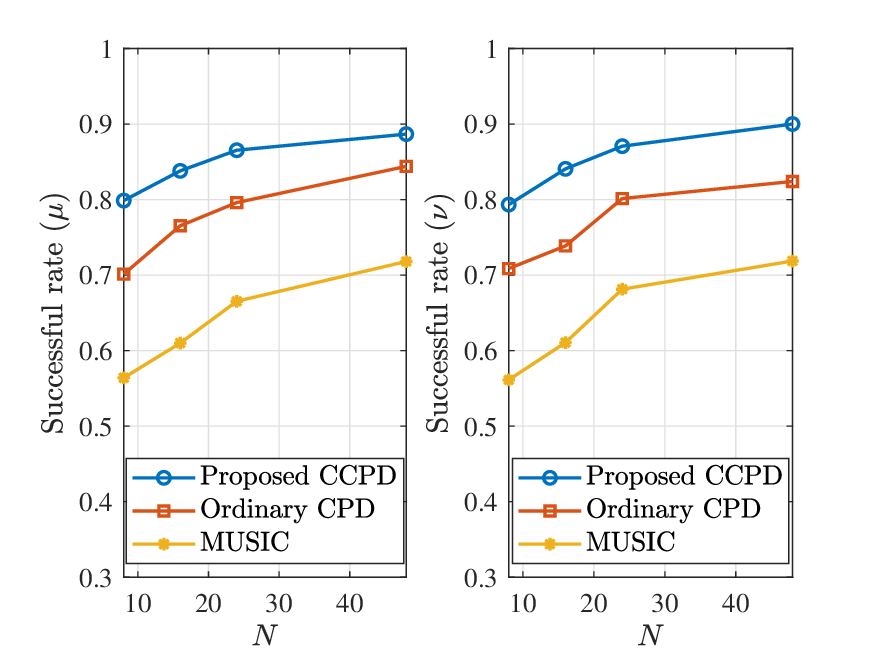}
	\caption{The success rate of correctly estimating $\mu$ and $\nu$ versus the number of OFDM symbols $N$, where $\rho = 45$ dBm and $K = 8$.}
	\label{Fig_success_N}
\end{figure}

Fig. \ref{Fig_success_rho} depicts the success rate of correctly estimating $\mu$ and $\nu$. 
A trial is considered successful if the estimation error satisfies $|\hat{\mu}_{q} - \mu_{q} | \leq \frac{1}{2 N_{\mathrm{bs},x}} $ or $|\hat{\nu}_{q} - \nu_{q} | \leq \frac{1}{2 N_{\mathrm{bs},y}} $, respectively \cite{Mamandipoor2016Newtonized}.
The simulation parameters are set to $K=8$ and $N=16$. 
From Fig. \ref{Fig_success_rho}, it can be observed that the success rate of the proposed CCPD method is considerably higher than that of the benchmark methods, and the proposed CCPD method performs similarly when estimating $\mu$ and $\nu$.
The MUSIC method achieves no more than an $80\%$ success rate even with a high $\rho = 50$ dBm, while the proposed CCPD method achieves a $90\%$ success rate.
Moreover, the proposed CCPD method exhibits a higher success rate compared with the ordinary CPD method, especially in the low-$\rho$ regime.
The significant enhancement of the success rate is attributed to the exploitation of the coupled CPD structure of factor matrices in \eqref{equ:CCPD}-\eqref{equ:C_est}, which verifies the theoretical analysis in Lemma \ref{lemma:coupled CPD}.

Fig. \ref{Fig_success_K} depicts the success rate of correctly estimating $\mu$ and $\nu$ versus the number of subcarriers $K$. 
The simulation parameters are set to $\rho = 45$ dBm and $N = 16$. 
From Fig. \ref{Fig_success_K}, the proposed CCPD method demonstrates a substantially higher success rate than the existing benchmark methods, and it achieves a higher success rate for $\nu$ than for $\mu$.
The MUSIC method and the ordinary CPD method achieve success rates of no more than $65\%$ and $80\%$, respectively, while the proposed CCPD method achieves a success rate approaching $90\%$. 
The success rate of the proposed CCPD method consistently improves as the number of subcarriers increases, demonstrating that exploiting the coupled structure across multiple subcarriers effectively enhances the robustness and resolution of parameter estimation.

Fig. \ref{Fig_success_N} depicts the success rate of correctly estimating $\mu$ and $\nu$ versus the number of OFDM symbols $N$. 
The simulation parameters are set to $\rho=45$ dBm and $K=8$. 
From Fig. \ref{Fig_success_N}, it can be observed that the proposed CCPD method maintains a higher success rate than the benchmark methods across different numbers of symbols, and it achieves a higher success rate for $\nu$ than for $\mu$, especially in the high-$N$ regime.
The ordinary CPD method exhibits a noticeable decline in the success rate when the number of symbols is lower than $24$, and the success rate of the MUSIC method drops sharply as the number of symbols decreases, while the success rate of the proposed CCPD method remains stable and consistently exceeds $80\%$. 
This shows that coupled CPD can effectively combine time-domain and frequency-domain information and improve the efficiency of time-resource utilization.

\begin{figure}[!tp]
	\centering
	\includegraphics[width=3.2in]{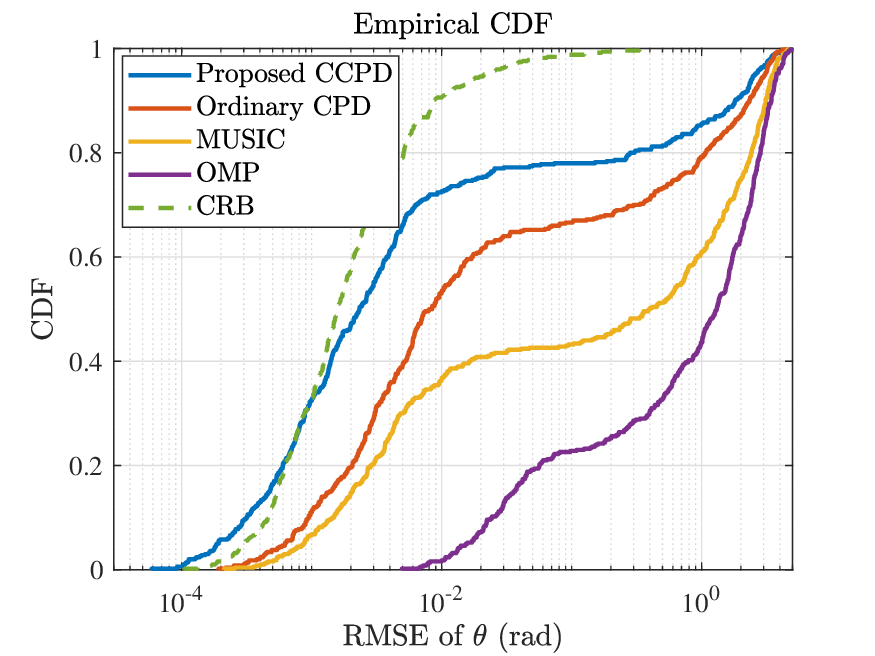}
	\caption{The CDF of RMSE performance with respect to azimuth angles $\theta$, where $\rho = 45$ dBm, $K = 8$, and $N = 32$.}
	\label{Fig_CDF_theta}
\end{figure}

\begin{figure}[!tp]
	\centering
	\includegraphics[width=3.2in]{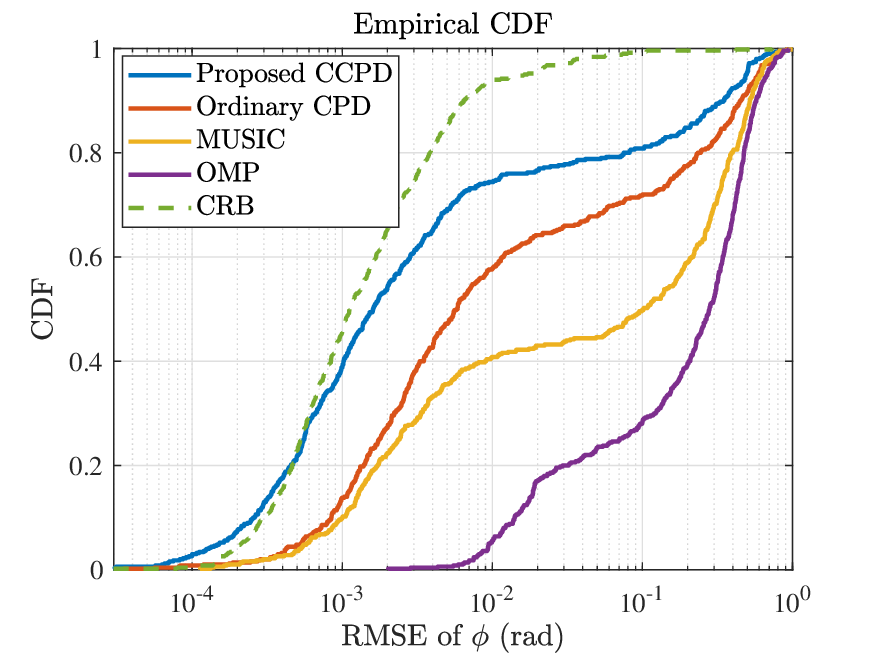}
	\caption{The CDF of RMSE performance with respect to elevation angles $\phi$, where $\rho = 45$ dBm, $K = 8$, and $N = 32$.}
	\label{Fig_CDF_phi}
\end{figure}

\begin{figure}[!tp]
	\centering
	\includegraphics[width=3.2in]{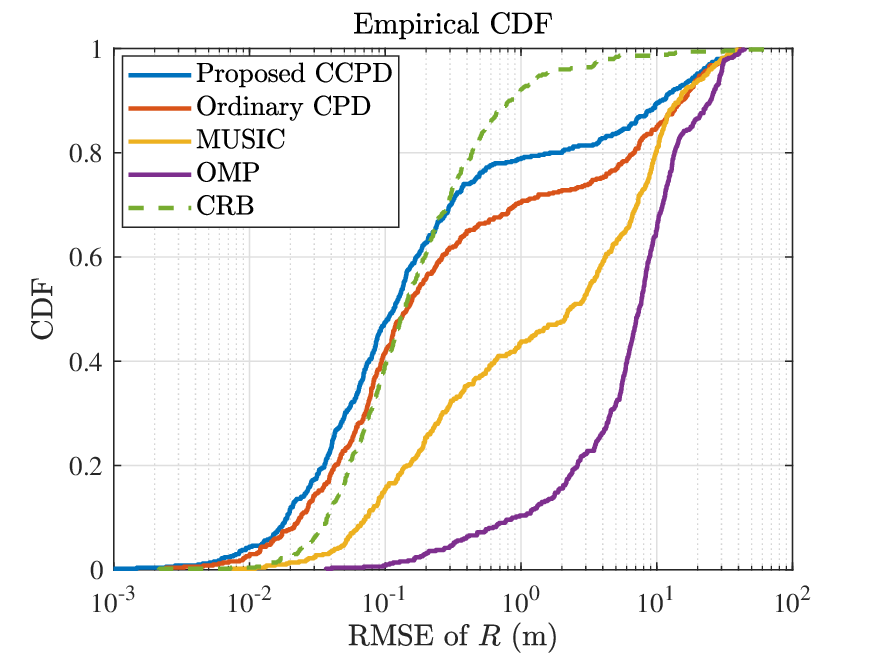}
	\caption{The CDF of RMSE performance with respect to range $R$, where $\rho = 45$ dBm, $K = 8$, and $N = 32$.}
	\label{Fig_CDF_R}
\end{figure}

\begin{figure}[!tp]
	\centering
	\includegraphics[width=3.2in]{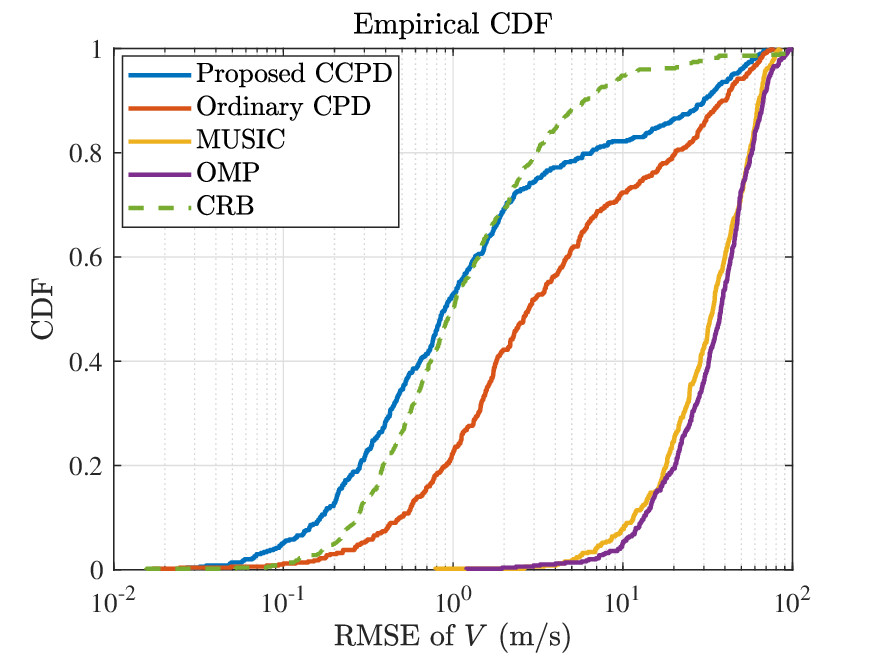}
	\caption{The CDF of RMSE performance with respect to velocity $V$, where $\rho = 45$ dBm, $K = 8$, and $N = 32$.}
	\label{Fig_CDF_V}
\end{figure}

Fig. \ref{Fig_CDF_theta} depicts the cumulative distribution function (CDF) of RMSE performance with respect to azimuth angles $\theta$, where $\rho = 45$ dBm, $K = 8$, and $N = 32$. 
For the CDF results, the empirical distributions are generated through the following Monte Carlo procedure. In each realization, the UE and target parameters are randomly generated according to Table~\ref{tab:simulation_parameters}. 
The proposed CCPD and benchmark methods are applied to parameter estimation based on the considered FD MIMO-OFDM ISAC model, and the RMSE is computed for each realization. 
After $N_{\rm MC}=500$ independent realizations, the empirical CDF is obtained by
$\mathrm{CDF}(x) = \frac{1}{N_{\mathrm{MC}}}
\sum_{n_{\mathrm{MC}}=1}^{N_{\mathrm{MC}}}
\mathbb{I}\left(\mathrm{RMSE}_{n_{\mathrm{MC}}} \le x\right)$, where $\mathbb{I}(\cdot)$ denotes the indicator function, i.e., $\mathbb{I}(A)=1$ if $A$ is true and $\mathbb{I}(A)=0$ otherwise.
For each Monte Carlo realization, the CRB is computed for the corresponding randomly generated target parameters, yielding $N_{\mathrm{MC}}$ CRB results for constructing the empirical CDF curve.
From Fig. \ref{Fig_CDF_theta}, it can be observed that the azimuth-angle estimation accuracy of the proposed CCPD method is higher than that of the ordinary CPD method, the MUSIC method, and the OMP method.
The results show that the coupled CPD framework can effectively suppress the parameter estimation error and improve the accuracy and consistency of angle estimation.
Fig. \ref{Fig_CDF_phi} depicts the CDF of RMSE performance with respect to elevation angles $\phi$ under the same simulation setting. The proposed CCPD method similarly exhibits a superior error distribution, further confirming the overall advantage of coupled CPD in two-dimensional angle estimation.

Fig. \ref{Fig_CDF_R} depicts the CDF of the range estimation RMSE under $\rho = 45$ dBm, $K = 8$, and $N = 32$.
It can be observed that the proposed CCPD method yields a more concentrated error distribution in the low-RMSE region than the benchmark methods.
In particular, its CDF increases rapidly when the RMSE of $R$ is below $10^{-1}$ m and remains closer to the CRB than the ordinary CPD, MUSIC, and OMP methods over most RMSE thresholds.
In contrast, the ordinary CPD method exhibits a heavier high-RMSE tail because it only exploits the echo signal tensor of a single subcarrier. 
Also, the MUSIC and OMP methods suffer from more pronounced error spreads in the medium-to-high RMSE region.
This performance gain demonstrates the improvement in accuracy and robustness of the proposed CCPD method in terms of range estimation.

Fig. \ref{Fig_CDF_V} depicts the CDF of the velocity estimation RMSE under $\rho = 45$ dBm, $K = 8$, and $N = 32$.
Compared with the benchmark methods, the proposed CCPD method achieves the most left-shifted CDF curve and remains closest to the CRB in the low-RMSE region. 
For a given CDF level, the RMSE achieved by the proposed method is clearly smaller than that achieved by the ordinary CPD, MUSIC, and OMP methods, indicating a higher probability of obtaining accurate velocity estimates.
The advantage over the ordinary CPD method comes from the joint exploitation of the common Doppler structure across subcarriers, whereas the MUSIC and OMP methods show a much slower CDF increase and heavier high-RMSE tails.
These observations verify that the proposed coupled tensor formulation improves velocity estimation performance from the communication data payload signals.

\begin{figure}[!tp]
	\centering
	\includegraphics[width=3.2in]{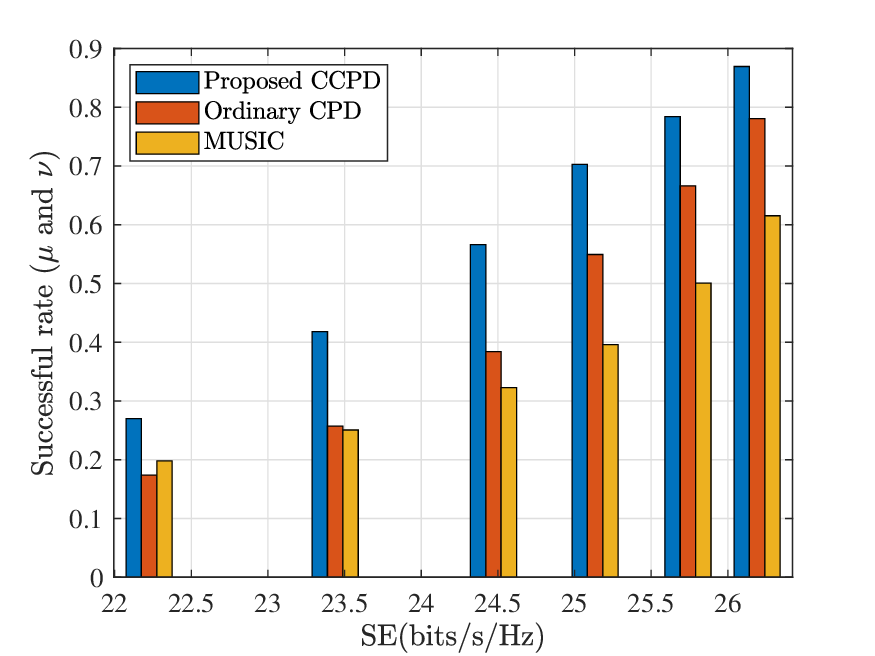}
	\caption{The success rate of correctly detecting targets as a function of SE, where $K=8$ and $N=16$.}
	\label{Fig_success_SE}
\end{figure}

Fig. \ref{Fig_success_SE} depicts the success rate of correctly detecting targets as a function of SE, where $K=8$ and $N=16$. One trial is considered successful if the estimation error satisfies both $|\hat{\mu}_{q} - \mu_{q} | \leq \frac{1}{2 N_{\mathrm{bs},x}} $ and $|\hat{\nu}_{q} - \nu_{q} | \leq \frac{1}{2 N_{\mathrm{bs},y}} $.
It can be observed that the proposed CCPD method maintains a high success rate, remaining above $40\%$ in a medium-SE regime, whereas the ordinary CPD method and the MUSIC method achieve only around a $25\%$ success rate under the same conditions.
Under the same spectral efficiency, the proposed CCPD method still maintains a high success rate, indicating that the proposed method provides more reliable target sensing while preserving communication performance.

To quantitatively evaluate the convergence behavior of the proposed CCPD method, the normalized reconstruction error (NRE) is adopted as the convergence metric. 
At the $i$-th iteration, the reconstructed tensor corresponding to the $k$-th subcarrier is expressed as 
$\widehat{\mathcal{Y}}_{k}^{(i)}
=
\llbracket
\mathbf{B}_{k}^{(i)},
\mathbf{C}_{k}^{(i)},
\mathbf{D}^{(i)}
\rrbracket$, 
where $\mathbf{B}_{k}^{(i)}$, $\mathbf{C}_{k}^{(i)}$, and
$\mathbf{D}^{(i)}$ denote the factor matrices estimated at the $i$-th CCPD iteration. The NRE is then defined as
\begin{align}
{\rm NRE}^{(i)}
    =
    \frac{
    \sum_{k=0}^{K-1}
    \left\|
    \mathcal{Y}_{k}
    -
    \widehat{\mathcal{Y}}_{k}^{(i)}
    \right\|_{F}^{2}
    }
    {
    \sum_{k=0}^{K-1}
    \left\|
    \mathcal{Y}_{k}
    \right\|_{F}^{2}
    },
\end{align}
where the numerator measures the total reconstruction error, while the denominator normalizes the error by the total energy of the received coupled tensors. Consequently, a smaller NRE indicates a more accurate approximation of the observed tensor signals.

\begin{figure}[!tp]
	\centering
	\includegraphics[width=3.2in]{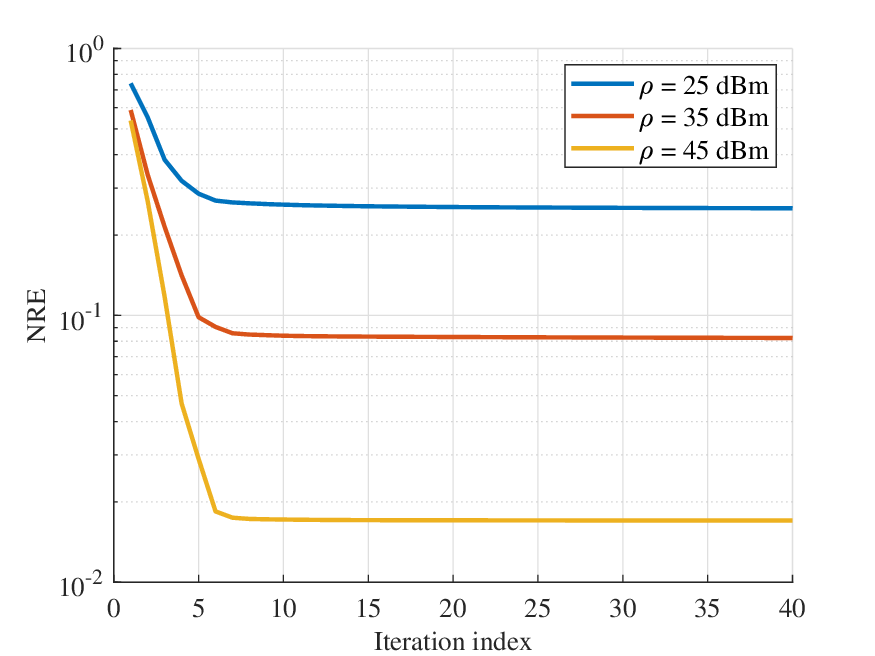}
	\caption{Normalized reconstruction error versus the number of CCPD iterations under different $\rho$.}
    \label{fig:conv_nre}
\end{figure}

Fig. \ref{fig:conv_nre} illustrates the convergence performance of the proposed CCPD method in terms of the NRE under different transmit power values $\rho$, where $K = 8$ and $N = 16$. It can be observed that the NRE decreases rapidly within the first few iterations and then gradually converges to a stable value within approximately 7 to 10 iterations, demonstrating the rapid convergence and numerical robustness of the proposed algorithm. Moreover, increasing the transmit power $\rho$ leads to a lower final NRE. This is because a larger transmit power improves the quality of the received echo signals, which enhances the accuracy of the coupled tensor decomposition and consequently reduces the reconstruction error. 
These results confirm that the proposed CCPD method can achieve effective convergence while maintaining high reconstruction accuracy.

\section{Conclusion}
In this paper, we proposed a novel target sensing method for FD MIMO-OFDM ISAC systems based on coupled tensor decomposition. The core of the proposed approach lies in formulating the target parameter estimation problem as a CCPD problem, which effectively leverages the communication data payload across multiple OFDM subcarriers. 
We proposed a two-stage CCPD-based target sensing algorithm, such that the target parameters, including azimuth and elevation angles, time delay, Doppler shift, and reflection coefficients, are separately extracted from the decomposed factor matrices. 
Theoretical analysis confirms that the proposed CCPD formulation yields a more relaxed uniqueness condition than the ordinary CPD method that only utilizes a single subcarrier. This enables more stable and reliable target parameter acquisition using pure communication data without dedicating additional resources or compromising communication functionality.
Numerical simulations demonstrate the performance superiority of the proposed method in terms of success rate and target-parameter estimation accuracy by exploiting the data payload signals.

\appendices

\section{PROOF OF LEMMA \ref{lemma:solution_c_q_k}}\label{proof_c_q_k}
\begin{proof}
    From \eqref{equ:c_q_k}, we have
    \begin{align}
        \hat{\mathbf{c}}_{q,k} 
        & = \delta_{q,k}^{(2)} \gamma_{q,k} \mathbf{d}_{y}(\nu_{q}) + \mathbf{e}_{q,k}^{(2)} \nonumber\\
        & = \tilde{\delta}_{q,k}^{(2)} \mathbf{d}_{y}(\nu_{q}) + \mathbf{e}_{q,k}^{(2)}, 
    \end{align}
    where $\nu_{q}$, $\gamma_{q,k}$, and $\delta_{q,k}^{(2)}$ are the unknown parameters, and the second equality holds by defining $\tilde{\delta}_{q,k}^{(2)}=\delta _{q,k}^{(2)}\gamma_{q,k}$. 
    Without loss of generality, we assume $\mathbf{e}_{q,k}^{(2)}\sim \mathcal{CN}(0,\sigma_{\mathrm{s}}^{2}\mathbf{I}_{M_{y}})$, and then the log-likelihood function for the unknown parameters is given by 
	\begin{align}
		L(\tilde{\delta}_{q,k}^{(2)},\nu _{q})
        =&\sum_{k=0}^{K-1}\Big(-M_{y}\ln(\pi \sigma _{\mathrm{s}}^{2})\notag\\
		&-\frac{1}{\sigma _{\mathrm{s}}^{2}}\left\| \hat{\mathbf{c}}_{q,k} - \tilde{\delta}_{q,k}^{(2)} \mathbf{d}_{y}(\nu _{q})\right\|_{2}^{2} \Big) \notag \\
		\propto& -\sum_{k=0}^{K-1}\left\| \hat{\mathbf{c}}_{q,k}-\tilde{\delta}_{q,k}^{(2)}\mathbf{d}_{y}(\nu _{q})\right\|_{2}^{2}. \label{equ:L_delta_nu}
	\end{align}

	Given a fixed $\nu_{q}$, by setting the partial derivative of the log-likelihood function with respect to $\tilde{\delta}_{q,k}^{(2)}$ to zero, i.e., $\frac{\partial L(\tilde{\delta}_{q,k}^{(2)},\nu _{q})}{\partial \tilde{\delta}_{q,k}^{(2)}}=0$, we have 
	\begin{align}
		\mathbf{d}_{y}^{H}(\nu_{q})(\hat{\mathbf{c}}_{q,k}-\tilde{\delta}_{q,k}^{(2)}\mathbf{d}_{y}(\nu_{q}))=0.
	\end{align}

	Thus, we can obtain the optimal solution of $\tilde{\delta}_{q,k}^{(2)}$ as follows 
	\begin{align}
		\hat{\tilde{\delta}}_{q,k}^{(2)}=\frac{\mathbf{d}_{y}^{H}(\nu _{q})\hat{\mathbf{c}}_{q,k}}{\left\| \mathbf{d}_{y}(\nu _{q})\right\|_{2}^{2}}. \label{equ:delta_q_k_'}
	\end{align}

	Substituting the optimal $\hat{\tilde{\delta}}_{q,k}^{(2)}$ from \eqref{equ:delta_q_k_'} into the log-likelihood function in \eqref{equ:L_delta_nu}, we have 
	\begin{align}
		L(\nu _{q})&\propto -\sum_{k=0}^{K-1}(\left\| \hat{\mathbf{c}}_{q,k}\right\|_{2}^{2}-\frac{\left | \mathbf{d}_{y}^{H}(\nu _{q})\hat{\mathbf{c}}_{q,k}\right |^{2}}{\left\| \mathbf{d}_{y}(\nu _{q})\right\|_{2}^{2}}) \notag \\
        & \propto -\sum_{k=0}^{K-1}\left\| \hat{\mathbf{c}}_{q,k}\right\|_{2}^{2}(1-\frac{\left | \mathbf{d}_{y}^{H}(\nu _{q})\hat{\mathbf{c}}_{q,k}\right |^{2}}{\left\| \mathbf{d}_{y}(\nu _{q})\right\|_{2}^{2}\left\| \hat{\mathbf{c}}_{q,k}\right\|_{2}^{2}})\notag \\
		&\propto -\sum_{k=0}^{K-1}(1-\frac{\left | \mathbf{d}_{y}^{H}(\nu _{q})\hat{\mathbf{c}}_{q,k}\right |^{2}}{\left\| \mathbf{d}_{y}(\nu _{q})\right\|_{2}^{2}\left\| \hat{\mathbf{c}}_{q,k}\right\|_{2}^{2}}).
	\end{align}

	Accordingly, the maximum likelihood solution of $\nu _{q}$ can be found by
	\begin{align}
		\hat{\nu}_{q} 
        = & \arg\max_{\nu _{q}}L(\nu _{q}) \notag \\
        = & \arg\max_{\nu _{q}}\sum_{k=0}^{K-1}\frac{\left | \mathbf{d}_{y}^{H}(\nu _{q})\hat{\mathbf{c}}_{q,k}\right |^{2}}{\left\| \mathbf{d}_{y}(\nu _{q})\right\|_{2}^{2}\left\| \hat{\mathbf{c}}_{q,k}\right\|_{2}^{2}}.
	\end{align}

    The proof is completed.
\end{proof}
    
\section{PROOF OF LEMMA \ref{lemma:solution_b_q_k}} \label{proof_b_q_k}
\begin{proof}
	From \eqref{equ:Sigma_a_theta}, we have
	\begin{align}
		\hat{\mathbf{b}}_{q,k} = \delta_{q,k}^{(1)} \mathbf{Q}_{k}^{T} \mathbf{a}_{\mathrm{do}}(\omega_{q})  + \mathbf{e}_{q}^{(1)},
	\end{align}
	where $\delta _{q,k}^{(1)}$ and $\omega _{q}$ are unknown parameters.
Similar to \eqref{equ:L_delta_nu}, the log-likelihood function of $\hat{\mathbf{b}}_{q,k}$ is given by 
	\begin{align}
		L(\delta _{q,k}^{(1)},\omega _{q})
		\propto& -\sum_{k=0}^{K-1}\left\| \hat{\mathbf{b}}_{q,k}-\delta _{q,k}^{(1)}\mathbf{Q}_{k}^{T}\mathbf{a}_{\mathrm{do}}(\omega _{q})\right\|_{2}^{2}. \label{equ:L_delta_omega}
	\end{align}

	Given a fixed $\omega _{q}$,
	by setting the partial derivative of the log-likelihood function with respect to $\delta _{q,k}^{(1)}$ to zero,
	i.e., $\frac{\partial L(\delta _{q,k}^{(1)},\omega _{q})}{\partial \delta _{q,k}^{(1)}}=0$,
	we have
	\begin{align}
		(\mathbf{Q}_{k}^{T}\mathbf{a}_{\mathrm{do}}(\omega_{q}))^{H}(\hat{\mathbf{b}}_{q,k}-\delta _{q,k}^{(1)}\mathbf{Q}_{k}^{T}\mathbf{a}_{\mathrm{do}}(\omega_{q}))=0. 
	\end{align}

	Thus, we can obtain the optimal $\delta _{q,k}^{(1)}$ as follows 
	\begin{align}
		\hat{\delta}_{q,k}^{(1)}=\frac{(\mathbf{Q}_{k}^{T}\mathbf{a}_{\mathrm{do}}(\omega_{q}))^{H}\hat{\mathbf{b}}_{q,k}}{\left\| \mathbf{Q}_{k}^{T}\mathbf{a}_{\mathrm{do}}(\omega_{q})\right\|_{2}^{2}} \label{equ:delta_q_k}.
	\end{align}

	Substituting the optimal $\delta _{q,k}^{(1)}$ from \eqref{equ:delta_q_k} into the log-likelihood function in \eqref{equ:L_delta_omega}, we have 
	\begin{align}
		L(\omega_{q})
		&\propto -\sum_{k=0}^{K-1}(1-\frac{\left | (\mathbf{Q}_{k}^{T}\mathbf{a}_{\mathrm{do}}(\omega_{q}))^{H}\hat{\mathbf{b}}_{q,k}\right |^{2}}{\left\| \mathbf{Q}_{k}^{T}\mathbf{a}_{\mathrm{do}}(\omega_{q})\right\|_{2}^{2}\left\| \hat{\mathbf{b}}_{q,k}\right\|_{2}^{2}}).
	\end{align}

	Thus, the maximum likelihood solution of $\omega_{q}$ can be found by
	\begin{align}
		\hat{\omega}_{q} 
		&= \arg\max_{\omega _{q}}\sum_{k=0}^{K-1}\frac{\left | (\mathbf{Q}_{k}^{T}\mathbf{a}_{\mathrm{do}}(\omega_{q}))^{H}\hat{\mathbf{b}}_{q,k}\right |^{2}}{\left\| \mathbf{Q}_{k}^{T}\mathbf{a}_{\mathrm{do}}(\omega_{q})\right\|_{2}^{2}\left\| \hat{\mathbf{b}}_{q,k}\right\|_{2}^{2}}.
	\end{align}

	The proof is completed.
\end{proof}



%

\bibliographystyle{IEEEtran}

\end{document}